\PassOptionsToPackage{unicode}{hyperref}
\PassOptionsToPackage{hyphens}{url}
\PassOptionsToPackage{dvipsnames,svgnames,x11names}{xcolor}
\documentclass[
  letterpaper,
  DIV=11,
  numbers=noendperiod]{scrartcl}
\usepackage{xcolor}
\usepackage{amsmath,amssymb}
\usepackage{iftex}
\ifPDFTeX
  \usepackage[T1]{fontenc}
  \usepackage[utf8]{inputenc}
  \usepackage{textcomp} % provide euro and other symbols
\else % if luatex or xetex
  \usepackage{unicode-math} % this also loads fontspec
  \defaultfontfeatures{Scale=MatchLowercase}
  \defaultfontfeatures[\rmfamily]{Ligatures=TeX,Scale=1}
\fi
\usepackage{lmodern}
\ifPDFTeX\else
\fi
\IfFileExists{upquote.sty}{\usepackage{upquote}}{}
\IfFileExists{microtype.sty}{% use microtype if available
  \usepackage[]{microtype}
  \UseMicrotypeSet[protrusion]{basicmath} % disable protrusion for tt fonts
}{}
\makeatletter
\@ifundefined{KOMAClassName}{% if non-KOMA class
  \IfFileExists{parskip.sty}{%
    \usepackage{parskip}
  }{% else
    \setlength{\parindent}{0pt}
    \setlength{\parskip}{6pt plus 2pt minus 1pt}}
}{% if KOMA class
  \KOMAoptions{parskip=half}}
\makeatother
\makeatletter
\ifx\paragraph\undefined\else
  \let\oldparagraph\paragraph
  \renewcommand{\paragraph}{
    \@ifstar
      \xxxParagraphStar
      \xxxParagraphNoStar
  }
  \newcommand{\xxxParagraphStar}[1]{\oldparagraph*{#1}\mbox{}}
  \newcommand{\xxxParagraphNoStar}[1]{\oldparagraph{#1}\mbox{}}
\fi
\ifx\subparagraph\undefined\else
  \let\oldsubparagraph\subparagraph
  \renewcommand{\subparagraph}{
    \@ifstar
      \xxxSubParagraphStar
      \xxxSubParagraphNoStar
  }
  \newcommand{\xxxSubParagraphStar}[1]{\oldsubparagraph*{#1}\mbox{}}
  \newcommand{\xxxSubParagraphNoStar}[1]{\oldsubparagraph{#1}\mbox{}}
\fi
\makeatother

\usepackage{longtable,booktabs,array}
\usepackage{calc} % for calculating minipage widths
\usepackage{etoolbox}
\makeatletter
\patchcmd\longtable{\par}{\if@noskipsec\mbox{}\fi\par}{}{}
\makeatother
\IfFileExists{footnotehyper.sty}{\usepackage{footnotehyper}}{\usepackage{footnote}}
\makesavenoteenv{longtable}
\usepackage{graphicx}
\makeatletter
\newsavebox\pandoc@box
\newcommand*\pandocbounded[1]{% scales image to fit in text height/width
  \sbox\pandoc@box{#1}%
  \Gscale@div\@tempa{\textheight}{\dimexpr\ht\pandoc@box+\dp\pandoc@box\relax}%
  \Gscale@div\@tempb{\linewidth}{\wd\pandoc@box}%
  \ifdim\@tempb\p@<\@tempa\p@\let\@tempa\@tempb\fi% select the smaller of both
  \ifdim\@tempa\p@<\p@\scalebox{\@tempa}{\usebox\pandoc@box}%
  \else\usebox{\pandoc@box}%
  \fi%
}
\def\fps@figure{htbp}
\makeatother

\providecommand{\tightlist}{%
  \setlength{\itemsep}{0pt}\setlength{\parskip}{0pt}}

\usepackage[backend=biber,style=numeric]{biblatex}
\DefineBibliographyStrings{english}{urlseen = {Visited on}}

\KOMAoption{captions}{tableheading}
\makeatletter
\@ifpackageloaded{caption}{}{\usepackage{caption}}
\AtBeginDocument{%
\ifdefined\contentsname
  \renewcommand*\contentsname{Table of contents}
\else
  \newcommand\contentsname{Table of contents}
\fi
\ifdefined\listfigurename
  \renewcommand*\listfigurename{List of Figures}
\else
  \newcommand\listfigurename{List of Figures}
\fi
\ifdefined\listtablename
  \renewcommand*\listtablename{List of Tables}
\else
  \newcommand\listtablename{List of Tables}
\fi
\ifdefined\figurename
  \renewcommand*\figurename{Figure}
\else
  \newcommand\figurename{Figure}
\fi
\ifdefined\tablename
  \renewcommand*\tablename{Table}
\else
  \newcommand\tablename{Table}
\fi
}
\@ifpackageloaded{float}{}{\usepackage{float}}
\floatstyle{ruled}
\@ifundefined{c@chapter}{\newfloat{codelisting}{h}{lop}}{\newfloat{codelisting}{h}{lop}[chapter]}
\floatname{codelisting}{Listing}

\usepackage{amsthm}
\theoremstyle{definition}
\newtheorem{definition}{Definition}
\theoremstyle{plain}
\newtheorem{theorem}{Theorem}
\theoremstyle{plain}
\newtheorem{lemma}{Lemma}
\theoremstyle{remark}
\AtBeginDocument{}

\makeatother
\makeatletter
\@ifpackageloaded{caption}{}{\usepackage{caption}}
\@ifpackageloaded{subcaption}{}{\usepackage{subcaption}}
\makeatother
\usepackage{bookmark}
\IfFileExists{xurl.sty}{\usepackage{xurl}}{} % add URL line breaks if available
\makeatletter
\@ifundefined{xmpquote}{\newcommand{\xmpquote}[1]{#1}}{}
\makeatother
\usepackage{hyperref}
\hypersetup{
  pdftitle={Detecting acceleration in global mean temperature using supralinear trends},
  pdfauthor={J. Eduardo Vera-Valdés},
  pdfkeywords={\xmpquote{Global warming
acceleration}, \xmpquote{Temperature trends}, \xmpquote{Quadratic
trends}, \xmpquote{Supralinear trends}, \xmpquote{Piecewise linear
trends}},
  colorlinks=true,
  linkcolor={blue},
  filecolor={Maroon},
  citecolor={Blue},
  urlcolor={Blue},
  pdfcreator={LaTeX via pandoc}}

\title{Detecting acceleration in global mean temperature using
supralinear trends}
\author{J. Eduardo Vera-Valdés\\
{\small ORCID: \href{https://orcid.org/0000-0002-0337-8055}{0000-0002-0337-8055}}\\
{\small Department of Mathematical Sciences, Aalborg University, 9220 Aalborg, Denmark}\\
{\small Center for Research in Energy: Economics and Markets (CoRE)}\\
{\small Corresponding author: \href{mailto:eduardo@math.aau.dk}{eduardo@math.aau.dk}}}
\date{}
\begin{document}
\maketitle
\begin{abstract}
Whether global mean surface temperature has accelerated has received
opposite answers in quick succession, reported as undetectable on data
to 2023 and as significant on data to 2024. This raises two questions:
since when is acceleration detectable, and how far does the answer
depend on the test used? We show that structural-break tests, often used
for this purpose, cannot reliably date the onset of a continuously
accelerating warming rate. We therefore test for curvature directly,
under specifications whose null is exactly a constant warming rate and
whose alternative is curvature rather than a break. One is a
linearithmic trend, which grows faster than linear but more slowly than
any polynomial of degree greater than one; the other estimates the
curvature freely from the data. We report alongside them the quadratic
trend used in the recent literature. Over 1970--2025, in a single
full-sample test, the linearithmic and quadratic coefficients are
positive and statistically significant in all five major temperature
datasets once El Niño and solar variability are controlled for. The
onset of acceleration is dated by re-estimating over expanding windows
and correcting for multiple testing. Under the linearithmic
specification, acceleration first becomes detectable in windows ending
in mid-2023 for NOAA, late 2023 for GISTEMP and mid-2024 for ERA5, and
not yet in HadCRUT or Berkeley. The free-curvature specification agrees
on direction but rejects in fewer datasets, since estimating the shape
costs precision. The linearithmic specification implies the least
curvature of the three, with an acceleration that declines over time,
and so the lowest warming of the three when extended to 2050.
\end{abstract}

\section{Introduction}\label{sec-introduction}

The statistical analysis of trends in global mean surface temperature
(GMST) has repeatedly turned on the choice of specification. Early
analyses established the presence of a warming trend against strongly
autocorrelated noise \autocite{Bloomfield1992}. A subsequent literature
asked whether that trend is deterministic or stochastic, and
\textcite{Kaufmann2010} showed that the conflicting answers then in
circulation could be traced to the statistical models used to reach
them. More recent work has moved beyond the mean, characterising warming
through trends in the whole distribution of temperature
\autocite{gadearivas_trends_2020}, and the econometric treatment of
climate series has since developed into a literature of its own
\autocite{clements_econometric_2024}. The question now at issue is not
whether GMST is warming but whether the rate of that warming is itself
increasing, and the pattern of specification-dependent disagreement has
repeated.

Whether GMST is accelerating has, accordingly, received sharply
different answers in quick succession. Using data up to 2023,
\textcite{beaulieu_recent_2024} concluded that a recent surge in global
warming was not yet detectable, and a parallel literature has
characterised GMST as stationary around a linear trend with a single
structural break placed in the 1970s, with no second break established
since \autocite{gadea-rivas_trends_2024,estrada_2013}. Shortly
afterwards, and with the record extended only to 2024,
\textcite{Forster2026} reported that GMST has accelerated with over 98\%
confidence. That analysis first filtered the data to control for
additional effects such as the El Niño effect and solar variation, and
then used quadratic and piecewise linear trend specifications to test
for acceleration in GMST.

Nearly the same observational record has therefore supported opposite
conclusions within a short span. Part of this is genuine new
information: 2023 and 2024 were exceptionally warm, and a handful of
observations at the end of a sample can carry substantial weight. But
part of it is the tests themselves, and the two sources are difficult to
separate without first asking what each specification is capable of
resolving. Two questions follow, and they are the ones this paper
addresses: since when is acceleration in GMST actually detectable, and
how far does the answer depend on the tool used to look for it? The
first is not only a statistical question: progress towards the
temperature goals of the Paris Agreement is assessed against how much
warming has occurred and how fast it is now proceeding
\autocite{forster_indicators_2025,vera-valdes2024}, so an increase in
the warming rate that goes undetected makes the time to reach those
goals appear longer than it is. The second is a question of robustness:
if the answer depends on the test used, then the claim that acceleration
is present is weaker and more prone to revision than if it is robust to
the choice of specification.

The two specifications used to date are informative, but neither is well
suited to the two questions above. A piecewise linear trend seems like a
direct method for dating a change in the warming rate: it allows a
discrete shift at an estimated date, and that date is precisely the
quantity of interest. We show in the Supplementary Material (SM,
Theorem~\ref{thm-breaks-supralinear-trend}) that a test for a structural
break cannot serve that purpose in general. If the true trend is
supralinear, or continuously increasing at a faster-than-linear rate,
then a break test will asymptotically reject the no-break null and
locate a break, even though nothing discrete has occurred in the
data-generating process. A rejection therefore establishes that the
warming rate is not constant, but it cannot establish that the change
was a discrete shift rather than continuous curvature: the two are
observationally close in any finite sample, and the date the test
returns need not mark an event. That break tests may locate breaks in
processes containing none is itself a long-standing concern in the
econometrics literature \autocite{Bai1998SpuriousBreak}. A quadratic
trend, for its part, guarantees a positive second derivative when the
quadratic coefficient is positive, but it imposes a rigid functional
form in which the warming rate increases at a constant rate. As Foster
and Rahmstorf note, ``the quadratic model, while a good choice to test
for acceleration, is a poor choice to estimate the present warming
rate'' \autocite{Forster2026}. Both have also been fitted at one sample
end point only, so they cannot answer the question of when acceleration
becomes detectable as the record accumulates. Dating the onset instead
requires the repeated application of a test whose null is a constant
warming rate and whose alternative is curvature, so that what is located
is the point at which curvature becomes detectable rather than the point
at which a break is best placed.

We therefore test for curvature directly, under two specifications that
between them leave the shape of the acceleration open rather than fixing
it in advance. The first is a linearithmic trend: \(t\log(t)\). It grows
faster than a linear trend but more slowly than any polynomial of degree
greater than one, so rejecting a constant warming rate in its favour
does not tie the result to a particular polynomial degree. The second
specification leaves the curvature exponent free and estimates it
jointly with its coefficient, nesting the linearithmic and quadratic
forms as special cases, and so serves as a check that the result is not
an artefact of a shape fixed in advance. The quadratic specification
used in the recent literature completes the design. On this sample its
curvature term and the linearithmic one are correlated at \(0.97\) once
the linear trend is removed, so the two fixed-shape tests are close to
the same test, and we report them side by side. We then ask when
acceleration becomes detectable by re-estimating over expanding windows
and correcting for multiple testing across them.

\section{Methods}\label{sec-methods}

\subsection{Trend specifications}\label{sec-specifications}

The literature has used quadratic trends to analyse acceleration in GMST
\autocite{Richardson2022Prospects,Forster2026}. A quadratic trend is
specified as follows:
\begin{equation}\protect\phantomsection\label{eq-quadratic-trend}{
x_t = \alpha + \lambda t + \beta t^2 + \varepsilon_t,
}\end{equation}

where \(x_t\) is the GMST anomaly at time \(t\), \(\alpha\) is the
intercept, \(\lambda\) is the slope of the linear trend, \(\beta\) is
the coefficient of the quadratic term, and \(\varepsilon_t\) is a random
error term. The presence of acceleration is tested by evaluating whether
\(\beta > 0\). If \(\beta = 0\), then there is no acceleration.

The motivation for using quadratic trends is that they allow for a
positive second derivative, which is the defining property of
acceleration. However, quadratic trends are not sufficient to capture
all possible forms of acceleration, as they impose a specific functional
form on the trend. A quadratic trend implies that the rate of change
increases at a constant rate over time, which may not be the case in
reality.

Hence, we propose a simple test for acceleration considering a
linearithmic specification of the trend function. Unlike a quadratic
trend, it does not require the warming rate to increase at a constant
rate, and unlike a piecewise linear trend it does not require the change
to occur at a single date. The linearithmic specification is given by:
\begin{equation}\protect\phantomsection\label{eq-linearithmic-trend}{
x_t = \alpha + \lambda t + \beta \log(t) t + \varepsilon_t,
}\end{equation}

where \(\alpha\) is the intercept, \(\lambda\) is the slope of the
linear trend, \(\beta\) is the coefficient of the linearithmic term, and
\(\varepsilon_t\) is a random error term. We show in the SM that
\(\hat\beta\) vanishes when the true trend is linear, converges to a
positive constant when it is linearithmic, and diverges when it is
polynomial of degree greater than one. Hence, the presence of
acceleration is tested by evaluating whether \(\beta > 0\).

The linearithmic specification has a key theoretical property. By
Lemma~\ref{lem-linearithmic-speed}, \(t\log(t)\) grows faster than any
linear trend but more slowly than any polynomial trend of degree greater
than one. Two things follow. First, the test is not tied to a particular
polynomial degree. Second, a rejection in favour of this specification
reads as a rising warming rate, without the constant rate of increase
that a quadratic trend implies. The difference is one of interpretation
more than of evidence. The curvature of \(t\log(t)\) declines over time,
and once the linear trend is partialled out its regressor is correlated
with \(t^2\) at \(0.969\) on this sample, so the linearithmic and
quadratic tests are close to the same test.

The linearithmic specification fixes the shape of the curvature in
advance, as the quadratic benchmark does. A natural robustness check is
to let the data choose the shape instead. We therefore also estimate a
free-curvature specification, in which the exponent governing the
curvature is a parameter to be estimated jointly with its coefficient.
The free-curvature specification is given by:
\begin{equation}\protect\phantomsection\label{eq-free-curvature}{
x_t = \alpha + \lambda t + \beta\; g(t,\delta) + \varepsilon_t; \quad
g(t,\delta) = \frac{t^{\delta} - t}{\delta - 1} \ \text{for} \ \delta \neq 1; \quad
g(t,1) := \lim_{\delta \to 1} g(t,\delta) = t\log(t).
}\end{equation}

The curvature function \(g(\cdot,\cdot)\) is written as a deviation from
the linear trend, which is what makes the exponent interpretable across
its whole range: \(\delta = 1\) returns the linearithmic specification
(Equation~\ref{eq-linearithmic-trend}), and \(\delta = 2\) returns
\(t^2 - t\), the quadratic curvature over and above a linear trend.
Equation~\ref{eq-free-curvature} therefore nests both of the preceding
specifications and tests the same null, \(H_0: \beta = 0\), which is
again exactly the hypothesis of a constant warming rate. We show in the
SM that acceleration corresponds to \(\beta > 0\) at every admissible
\(\delta\), so the test rests on \(\hat\beta\) alone. The exponent,
\(\hat\delta\), is reported as a shape estimate, describing how the
curvature accumulates.

In what follows, all three specifications are used for testing an
increase in the warming rate rather than for describing the trend.
Rejecting \(H_0: \beta = 0\) under any of them establishes that the
warming rate is not constant; it does not establish that GMST follows a
linearithmic, a quadratic or any other particular path, and nothing
below should be read as claiming that it does. The distinction matters
most outside the sample, where the forms diverge sharply, and we return
to it in Section~\ref{sec-beyond-sample}.

\subsection{Controls and estimation}\label{sec-controls}

We include the El Niño effect and solar variation as additional control
variables in all considered specifications. The control variables are
included directly in a single-step estimation, rather than through a
two-step adjustment, which allows for joint inference on the trend and
the controls. The inclusion of these controls is motivated by the
literature, which has shown that both El Niño and solar variation can
have significant effects on GMST
\autocite{le08200o,FosterRahmstorf2011,Forster2026}.

Let \(N_t\) denote the El Niño effect at time \(t\), measured by the
Oceanic Niño Index (ONI), and let \(S_t\) denote a measure of solar
irradiance at time \(t\), for which we use the sunspot number. Total
solar irradiance (TSI) is the more direct measure of solar forcing, but
over much of the sample it is a reconstruction in which the sunspot
number is a principal input, so using it would add the reconstruction's
modelling assumptions without adding information. The solar coefficient
is a control that we do not interpret, so the choice of proxy does not
affect any reported inference.

The El Niño effect enters the specifications with a lag, as is standard
in the literature \autocite{le08200o,FosterRahmstorf2011,Forster2026}.
The lag also has a mechanical motivation. The ONI is a centred
three-month running mean of sea surface temperature anomalies, so its
value at month \(t\) includes temperatures from month \(t+1\), over a
region that is itself part of the global mean. We select the lag by
goodness of fit over \(k = 0, \ldots, 8\) months. Every specification
considered in this paper selects \(k = 3\) months, unanimously across
the five datasets, which dates the regressor strictly before \(t\). The
SM reports the selection in full and shows that the estimate of
\(\beta\), and hence the evidence on acceleration, is essentially
invariant to \(k\).

The three specifications with additional control variables are then
given by:
\begin{equation}\protect\phantomsection\label{eq-quadratic-trend-exogenous}{
x_t = \alpha + \lambda t + \beta t^2 + \gamma_1 N_{t-k} + \gamma_2 S_t + \varepsilon_t;
}\end{equation}
\begin{equation}\protect\phantomsection\label{eq-linearithmic-trend-exogenous}{
x_t = \alpha + \lambda t + \beta \log(t) t + \gamma_1 N_{t-k} + \gamma_2 S_t + \varepsilon_t;
}\end{equation}
\begin{equation}\protect\phantomsection\label{eq-free-curvature-exogenous}{
x_t = \alpha + \lambda t + \beta\, g(t,\delta) + \gamma_1 N_{t-k} + \gamma_2 S_t + \varepsilon_t.
}\end{equation}

Above, \(\gamma_1\) and \(\gamma_2\) are the coefficients for the El
Niño effect and solar variation. The first two specifications are
estimated by ordinary least squares (OLS) and the third by non-linear
least squares (NLS) under the constraint \(\delta \in [1, 10]\), whose
lower bound has the interpretable reading \emph{at least linearithmic};
see the SM for technical details. In every case we use
heteroskedasticity and autocorrelation consistent (HAC) standard errors,
computed with a Bartlett kernel and Newey-West bandwidth selection
\autocite{newey-west-1987,NeweyWest1994}, to account for serial
correlation in the error term. On the full sample the rule selects a
bandwidth of \(15\) to \(18\) months, so the long-run variance is formed
from autocovariances out to roughly a year and a half.

\subsection{Dating the onset}\label{sec-windows}

We estimate the three specifications under two designs. The first fits
them to the full sample from 1970 to 2025; it establishes whether
acceleration is present as of the last observation considered, and puts
the out-of-sample implications of the three specifications on a common
basis. The question of \emph{when} acceleration becomes detectable
requires re-estimating the specifications as the record accumulates, and
we therefore deploy an expanding-window design. Starting with the
subsample from 1970 to 1990, we test \(H_0: \beta = 0\) in each of the
three specifications, and repeat the procedure expanding the end point
of the estimation window month by month to 2025.

The starting point of 1970 aligns with the literature on the observed
increase in the warming rate since the 1970s
\autocite{gadea-rivas_trends_2024,estrada_2013}. The first end point is
set to 1990 so that the sample is long enough for the bootstrap
procedure described below to be informative, and the last to 2025 so as
to include the most recent data available. Because \(t\log(t)\) is not
shift-invariant, this choice is a modelling assumption rather than a
normalisation; the SM re-estimates from anchors between 1955 and 1980
and shows that the set of rejecting datasets depends on this choice.

This design tests the same null across 432 nested windows, and the
quadratic and piecewise linear specifications have themselves been
applied at several points in the literature, so the nominal size of any
single test understates the probability of at least one spurious
rejection. Reporting the first window in which an uncorrected
\(p\)-value falls below \(5\%\) would therefore date the onset too early
by construction. To control the family-wise error rate across all
windows simultaneously, we compute critical values by a max-\(|t|\)
block wild bootstrap, following the resampling-based logic of
\textcite{westfall1993resampling}. To preserve the serial correlation of
the residuals, we multiply them by Rademacher random variables that are
held constant within blocks of ten consecutive months and independent
across blocks, forming bootstrap pseudo-samples. The test statistic is
then computed on each pseudo-sample, and the empirical distribution of
the maximum absolute value of the test statistic across all windows is
used to determine the critical value for the family-wise error rate. The
block length is selected by calibrating the actual size of the procedure
under a linear null with ARMA(1,1) errors fitted to the residuals. The
SM gives the algorithm and the calibration. All \(p\)-values reported
below as adjusted are family-wise across the full set of windows.

\section{Results}\label{sec-results}

We test for acceleration in GMST using the five most widely used
temperature series: HadCRUT \autocite{HadCRUT5}, GISTEMP
\autocite{GISTEMP}, NOAA \autocite{NOAA}, Berkeley Earth
\autocite{BerkeleyEarth}, and ERA5 \autocite{ERA5}. The full estimation
window runs from 1970-01 to 2025-12.

Throughout, \(t\) indexes months since the start of the estimation
window, so that \(t = 1\) corresponds to 1970-01 and \(t = T = 672\) to
2025-12. The trend regressors are scaled by \(T\), so \(\lambda\) and
\(\beta\) are expressed in °C per unit of the rescaled regressor rather
than in °C per month. The interpretable quantity is the implied warming
rate, which we report in °C decade\(^{-1}\) where relevant. As noted in
Section~\ref{sec-windows}, the 1970 anchor is a substantive modelling
assumption rather than a normalisation, since \(t\log(t)\) is not
invariant to a shift of the time origin; we also report in the SM the
results from re-estimating from anchors between 1955 and 1980.

\subsection{Acceleration over the full sample}\label{sec-full-sample}

Results for the three specifications in the full sample are presented in
Table~\ref{tbl-results-full-sample}. Following \textcite{Forster2026},
the specifications reported in the main text all include the El Niño and
solar controls. The corresponding estimates without them are reported in
the SM, and are summarised at the end of this subsection.

\begin{longtable}[]{@{}
  >{\raggedright\arraybackslash}p{(\linewidth - 14\tabcolsep) * \real{0.2000}}
  >{\raggedleft\arraybackslash}p{(\linewidth - 14\tabcolsep) * \real{0.1143}}
  >{\raggedleft\arraybackslash}p{(\linewidth - 14\tabcolsep) * \real{0.1143}}
  >{\raggedleft\arraybackslash}p{(\linewidth - 14\tabcolsep) * \real{0.1143}}
  >{\raggedleft\arraybackslash}p{(\linewidth - 14\tabcolsep) * \real{0.1143}}
  >{\raggedleft\arraybackslash}p{(\linewidth - 14\tabcolsep) * \real{0.1143}}
  >{\raggedleft\arraybackslash}p{(\linewidth - 14\tabcolsep) * \real{0.1143}}
  >{\raggedleft\arraybackslash}p{(\linewidth - 14\tabcolsep) * \real{0.1143}}@{}}
\caption{Estimation results for the linearithmic, quadratic and
free-curvature specifications, each estimated with the El Niño and solar
controls on 1970--2025 (Equation~\ref{eq-quadratic-trend-exogenous};
Equation~\ref{eq-linearithmic-trend-exogenous};
Equation~\ref{eq-free-curvature-exogenous}). Standard errors, calculated
using a HAC estimator, are shown in parentheses, and \(t_\beta\) is the
HAC \(t\)-statistic for \(H_0:\beta=0\), before the size correction of
Table~\ref{tbl-size-corrected}. The exponent \(\delta\) applies only to
the free-curvature specification, which is estimated by NLS. † marks an
exponent resting on the upper constraint \(\delta = 10\), for which no
standard error is defined. The corresponding estimates without controls
are reported in the SM.}\label{tbl-results-full-sample}\tabularnewline
\toprule\noalign{}
\begin{minipage}[b]{\linewidth}\raggedright
\textbf{Dataset / Model}
\end{minipage} & \begin{minipage}[b]{\linewidth}\raggedleft
\(\alpha\)
\end{minipage} & \begin{minipage}[b]{\linewidth}\raggedleft
\(\beta\)
\end{minipage} & \begin{minipage}[b]{\linewidth}\raggedleft
\(t_\beta\)
\end{minipage} & \begin{minipage}[b]{\linewidth}\raggedleft
\(\gamma_1\)
\end{minipage} & \begin{minipage}[b]{\linewidth}\raggedleft
\(\gamma_2\)
\end{minipage} & \begin{minipage}[b]{\linewidth}\raggedleft
\(\lambda\)
\end{minipage} & \begin{minipage}[b]{\linewidth}\raggedleft
\(\delta\)
\end{minipage} \\
\midrule\noalign{}
\endfirsthead
\toprule\noalign{}
\begin{minipage}[b]{\linewidth}\raggedright
\textbf{Dataset / Model}
\end{minipage} & \begin{minipage}[b]{\linewidth}\raggedleft
\(\alpha\)
\end{minipage} & \begin{minipage}[b]{\linewidth}\raggedleft
\(\beta\)
\end{minipage} & \begin{minipage}[b]{\linewidth}\raggedleft
\(t_\beta\)
\end{minipage} & \begin{minipage}[b]{\linewidth}\raggedleft
\(\gamma_1\)
\end{minipage} & \begin{minipage}[b]{\linewidth}\raggedleft
\(\gamma_2\)
\end{minipage} & \begin{minipage}[b]{\linewidth}\raggedleft
\(\lambda\)
\end{minipage} & \begin{minipage}[b]{\linewidth}\raggedleft
\(\delta\)
\end{minipage} \\
\midrule\noalign{}
\endhead
\bottomrule\noalign{}
\endlastfoot
\textbf{Berkeley} & & & & & & & \\
Linearithmic & 0.2122 & 0.2933 & 2.85 & 0.0892 & 0.0005 & -0.7952 & \\
& (0.0323) & (0.1028) & & (0.0079) & (0.0001) & (0.6765) & \\
Quadratic & 0.1869 & 0.0006 & 3.39 & 0.0887 & 0.0005 & 0.7373 & \\
& (0.0227) & (0.0002) & & (0.0079) & (0.0001) & (0.1205) & \\
Free-curvature & 0.1619 & 2.2311 & 1.23 & 0.0870 & 0.0003 & 1.2910 &
10.00\textsuperscript{†} \\
& (0.0194) & (1.8171) & & (0.0084) & (0.0001) & (0.0874) & --- \\
\textbf{ERA5} & & & & & & & \\
Linearithmic & 0.2300 & 0.4362 & 4.34 & 0.0996 & 0.0007 & -1.7029 & \\
& (0.0342) & (0.1005) & & (0.0088) & (0.0001) & (0.6604) & \\
Quadratic & 0.1879 & 0.0009 & 4.79 & 0.0988 & 0.0006 & 0.6037 & \\
& (0.0291) & (0.0002) & & (0.0086) & (0.0001) & (0.1266) & \\
Free-curvature & 0.1597 & 1.2309 & 1.95 & 0.0969 & 0.0006 & 1.2920 &
4.6915 \\
& (0.0267) & (0.6301) & & (0.0087) & (0.0001) & (0.0618) & (2.1788) \\
\textbf{GISTEMP} & & & & & & & \\
Linearithmic & 0.2763 & 0.4204 & 4.22 & 0.0882 & 0.0005 & -1.6674 & \\
& (0.0317) & (0.0995) & & (0.0081) & (0.0001) & (0.6550) & \\
Quadratic & 0.2355 & 0.0009 & 4.94 & 0.0875 & 0.0005 & 0.5569 & \\
& (0.0217) & (0.0002) & & (0.0080) & (0.0001) & (0.1123) & \\
Free-curvature & 0.2110 & 0.9982 & 2.36 & 0.0859 & 0.0004 & 1.2028 &
3.9986 \\
& (0.0231) & (0.4230) & & (0.0093) & (0.0001) & (0.0481) & (1.6705) \\
\textbf{HadCRUT} & & & & & & & \\
Linearithmic & 0.2455 & 0.3110 & 3.17 & 0.0871 & 0.0005 & -0.9240 & \\
& (0.0321) & (0.0982) & & (0.0083) & (0.0001) & (0.6491) & \\
Quadratic & 0.2159 & 0.0006 & 3.66 & 0.0865 & 0.0005 & 0.7182 & \\
& (0.0226) & (0.0002) & & (0.0083) & (0.0001) & (0.1155) & \\
Free-curvature & 0.2008 & 0.6385 & 1.84 & 0.0857 & 0.0004 & 1.1880 &
3.3851 \\
& (0.0250) & (0.3472) & & (0.0090) & (0.0001) & (0.0499) & (1.8967) \\
\textbf{NOAA} & & & & & & & \\
Linearithmic & 0.2178 & 0.4233 & 4.24 & 0.0846 & 0.0004 & -1.7287 & \\
& (0.0340) & (0.0998) & & (0.0074) & (0.0001) & (0.6615) & \\
Quadratic & 0.1792 & 0.0009 & 5.50 & 0.0839 & 0.0004 & 0.4952 & \\
& (0.0214) & (0.0002) & & (0.0072) & (0.0001) & (0.1076) & \\
Free-curvature & 0.1544 & 1.0646 & 2.20 & 0.0823 & 0.0003 & 1.1653 &
4.0755 \\
& (0.0209) & (0.4836) & & (0.0076) & (0.0001) & (0.0565) & (1.6728) \\
\end{longtable}

The estimated linearithmic coefficient is positive in every dataset.
Given the persistence of the errors, we refer each statistic to a null
distribution simulated under ARMA(1,1) errors fitted to that dataset's
own residuals, which sets the critical values at \(2.5\) to \(2.6\) at
the \(5\%\) level and \(3.4\) to \(3.5\) at the \(1\%\), against the
standard \(1.96\) and \(2.58\) (Table~\ref{tbl-size-corrected} and the
SM give the simulation in full). On that basis the coefficient is
significant in all five datasets, with size-corrected \(p\)-values:
GISTEMP (\(0.0018\)), ERA5 (\(0.0022\)) and NOAA (\(0.0023\)) reject at
the \(1\%\) level, and HadCRUT (\(0.0163\)) and Berkeley (\(0.0314\)) at
the \(5\%\). The quadratic specification gives the same answer: its
coefficient is positive in every dataset and, size-corrected in the same
way, significant at the \(5\%\) level in all five, and at the \(1\%\)
level in all but Berkeley (\(0.0103\)).

The free-curvature specification agrees on sign and differs in
precision. Its curvature coefficient is positive in every dataset, and
the estimated exponent is well above the quadratic value of two wherever
it falls in the interior of the admissible range (from \(3.39\) for
HadCRUT to \(4.69\) for ERA5), and at the upper constraint of ten in
Berkeley. Its standard errors on \(\hat\beta\) are several times larger
than in the linearithmic and quadratic specifications, because
\(\delta\) is estimated jointly with \(\beta\). Significance is
therefore not comparable between the OLS and the NLS rows on this basis,
and we defer it to the bootstrap-adjusted comparison of
Section~\ref{sec-onset}, which applies the same procedure to all three
specifications.

The controls are necessary to detect acceleration. Estimated without
them, the linearithmic and quadratic coefficients reach the 5\% level in
only three of the five datasets, and the free-curvature coefficient
reaches it in none. More decisively, once the correction for multiple
testing developed below is applied, not one of the \(2{,}160\)
window-level tests rejects in either OLS specification, the smallest
adjusted \(p\)-value being \(0.377\) for the linearithmic specification
and \(0.281\) for the quadratic. Nominal significance in the full sample
therefore does not survive the correction when the controls are absent.
This is consistent with recent findings that such controls are needed to
detect significant acceleration in GMST \autocite{Forster2026}, and it
is why the main text reports the specifications with controls
throughout; the SM gives the estimates in full.

The coefficients in Table~\ref{tbl-results-full-sample} are expressed in
units of the rescaled trend regressors and are not directly
interpretable. They are also not comparable across specifications, since
\(\beta\) multiplies a different function of time in each.
Table~\ref{tbl-implied-rates} therefore translates them into the
quantity that is interpretable: the implied warming rate in °C
decade\(^{-1}\), evaluated at three points of the sample, together with
the acceleration in °C decade\(^{-2}\) at the end of the
sample. The latter is the direct translation of \(\beta\).

\begin{longtable}[]{@{}
  >{\raggedright\arraybackslash}p{(\linewidth - 8\tabcolsep) * \real{0.3684}}
  >{\raggedleft\arraybackslash}p{(\linewidth - 8\tabcolsep) * \real{0.1579}}
  >{\raggedleft\arraybackslash}p{(\linewidth - 8\tabcolsep) * \real{0.1579}}
  >{\raggedleft\arraybackslash}p{(\linewidth - 8\tabcolsep) * \real{0.1579}}
  >{\raggedleft\arraybackslash}p{(\linewidth - 8\tabcolsep) * \real{0.1579}}@{}}
\caption{Implied warming rate (°C decade\(^{-1}\)) at three dates in the
sample, and acceleration (°C decade\(^{-2}\)) at the end of the sample,
under each specification estimated on 1970--2025 with controls. Each
cell gives the range across the five temperature datasets; the
free-curvature row covers the four datasets whose exponent is interior,
Berkeley being pinned at the upper constraint. The linear specification
is the null of the tests, and holds its rate fixed by construction.
Per-dataset values are available in the manuscript's
repository.}\label{tbl-implied-rates}\tabularnewline
\toprule\noalign{}
\begin{minipage}[b]{\linewidth}\raggedright
\textbf{Specification}
\end{minipage} & \begin{minipage}[b]{\linewidth}\raggedleft
Rate, 1980
\end{minipage} & \begin{minipage}[b]{\linewidth}\raggedleft
Rate, 2000
\end{minipage} & \begin{minipage}[b]{\linewidth}\raggedleft
Rate, 2025
\end{minipage} & \begin{minipage}[b]{\linewidth}\raggedleft
Acceleration, 2025
\end{minipage} \\
\midrule\noalign{}
\endfirsthead
\toprule\noalign{}
\begin{minipage}[b]{\linewidth}\raggedright
\textbf{Specification}
\end{minipage} & \begin{minipage}[b]{\linewidth}\raggedleft
Rate, 1980
\end{minipage} & \begin{minipage}[b]{\linewidth}\raggedleft
Rate, 2000
\end{minipage} & \begin{minipage}[b]{\linewidth}\raggedleft
Rate, 2025
\end{minipage} & \begin{minipage}[b]{\linewidth}\raggedleft
Acceleration, 2025
\end{minipage} \\
\midrule\noalign{}
\endhead
\bottomrule\noalign{}
\endlastfoot
Linear (constant rate) & 0.196--0.216 & 0.196--0.216 & 0.196--0.216 &
0 \\
Linearithmic & 0.129--0.162 & 0.212--0.233 & 0.251--0.281 &
0.009--0.014 \\
Quadratic & 0.127--0.159 & 0.204--0.224 & 0.283--0.324 & 0.027--0.039 \\
Free-curvature & 0.148--0.172 & 0.184--0.201 & 0.326--0.451 &
0.069--0.184 \\
\end{longtable}

Under the linearithmic specification the warming rate rises from
\(0.13\)--\(0.16\) °C decade\(^{-1}\) around 1980 to \(0.25\)--\(0.28\)
°C decade\(^{-1}\) at the end of the sample, with the final value in
line with recent estimates of the present rate
\autocite{forster_indicators_2025,jarvis2024estimated}. The comparison
with the constant-rate row is what the test is about: in 1980 all three
accelerating specifications place the rate below the constant-rate
estimate and by 2025 they place it above, which is precisely what a
rising rate implies for a linear trend fitted to the whole sample. The
three specifications agree closely in the middle of the sample and
separate at its ends, in the order their functional forms dictate. This
is the in-sample counterpart of the divergence beyond the sample
discussed in Section~\ref{sec-beyond-sample}. The two fixed-shape
specifications also differ in how the acceleration evolves. The
curvature of \(t\log(t)\) falls as \(1/t\), so the acceleration implied
by the linearithmic specification declines over time. At the end of the
sample it is \(0.009\) to \(0.014\) °C decade\(^{-2}\), about a third of
the quadratic's \(0.027\) to \(0.039\).

The estimated linear slope \(\hat{\lambda}\) is negative in every
dataset in the linearithmic specification, and positive in the quadratic
and free-curvature ones. Neither sign should be read as a warming rate,
and the contrast between them is bookkeeping rather than disagreement.
Where two trend regressors are present neither coefficient is separately
interpretable, and the sign of \(\hat\lambda\) depends on where the
curvature regressor is anchored: writing \(x = t/T\), the identity
\(x\log(x) = x\log(t) - x\log(T)\) shows that, at \(\delta = 1\), the
linearithmic and free-curvature parametrizations describe the same trend
with values of \(\lambda\) differing by \(\beta\log(T)\), so the sign of
\(\hat\lambda\) depends on the basis in which the curvature is written,
which is what carries it across zero. The interpretable quantity is the
implied rate \(\hat{\lambda} + \hat{\beta}(1 + \log t)\), which is
increasing throughout and which Table~\ref{tbl-implied-rates} reports in
°C decade\(^{-1}\). Under the linearithmic specification this rate is
negative at the start of the sample, over the first five to twenty-one
months depending on the dataset, which is why it starts below zero in
panel (a) of Figure~\ref{fig-projections}. A negative warming rate in
the early 1970s is not a feature of the data but of the functional form,
whose curvature is concentrated at the start of the sample. Individually
the two trend coefficients are close to unidentified; it is the implied
rate, and its increase, that the data determine sharply.

In summary, the three specifications agree on the presence and the
direction of curvature despite imposing very different amounts of
structure on it. The results show that the warming rate is not constant
and that it has increased over the past five decades. The linearithmic
specification implies the smallest increase in the rate, and one that
slows over time, while the quadratic and free-curvature specifications
imply larger increases. The inclusion of controls for El Niño and solar
variation is crucial for detecting acceleration, as their omission leads
to a failure to reject the null hypothesis of a constant warming rate in
several of the estimated specifications.

Figure~\ref{fig-fitted-values} presents the actual and fitted values for
the three specifications.

\begin{figure}

\begin{minipage}[t]{0.33\linewidth}

\raisebox{-\height}{

\includegraphics[width=1\linewidth,height=\textheight,keepaspectratio,alt={Linearithmic fit.}]{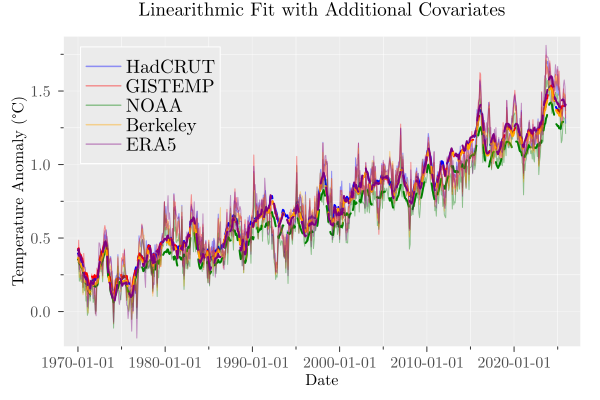}

}

\subcaption{\label{}Linearithmic fit.}
\end{minipage}%
\begin{minipage}[t]{0.33\linewidth}

\raisebox{-\height}{

\includegraphics[width=1\linewidth,height=\textheight,keepaspectratio,alt={Quadratic fit.}]{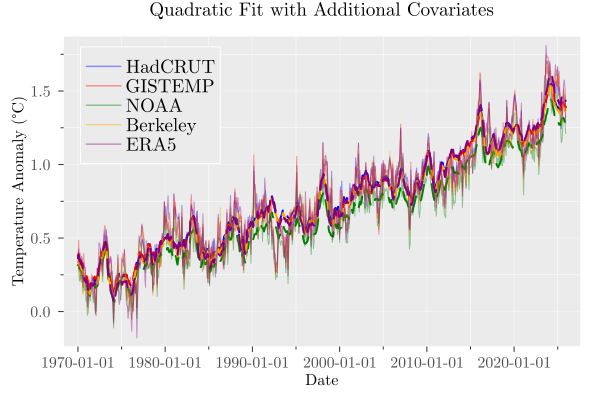}

}

\subcaption{\label{}Quadratic fit.}
\end{minipage}%
\begin{minipage}[t]{0.33\linewidth}

\raisebox{-\height}{

\includegraphics[width=1\linewidth,height=\textheight,keepaspectratio,alt={Free-curvature fit.}]{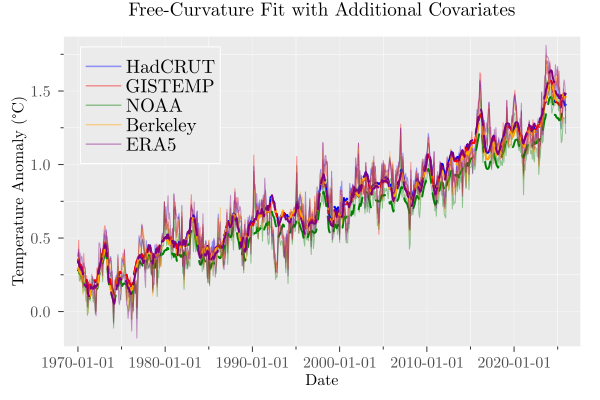}

}

\subcaption{\label{}Free-curvature fit.}
\end{minipage}%

\caption{\label{fig-fitted-values}Actual and fitted values for the
linearithmic, quadratic and free-curvature specifications from 1970 to
2025, each estimated with the El Niño and solar controls.}

\end{figure}%

The figure shows that all three specifications capture the overall trend
in GMST and provide similar fitted values. On formal fit criteria the
quadratic is preferred. Its adjusted \(R^2\) is higher in every dataset,
at \(0.887\) to \(0.906\) against \(0.885\) to \(0.903\) for the
linearithmic model, and its BIC is lower by \(8\) to \(22\), which is
strong evidence in its favour. Both improve on a constant warming rate
by far more, by \(30\) to \(84\) BIC points for the linearithmic and
\(41\) to \(106\) for the quadratic, so the evidence that the warming
rate is not constant is stronger still than the evidence for the
quadratic shape. Results in full are shown in the SM. Where the
specifications disagree is out of sample, since a quadratic trend
implies a warming rate that grows without bound at a constant rate,
while a linearithmic trend implies one that grows logarithmically; we
return to this in Section~\ref{sec-beyond-sample}.

\subsection{When does acceleration become detectable?}\label{sec-onset}

To be able to date when acceleration becomes detectable, we now turn to
the expanding-window design. Figure~\ref{fig-bootstrap-pvals} presents
the bootstrap-adjusted \(p\)-values for the test of acceleration for all
datasets and specifications considered.

\begin{figure}

\begin{minipage}[t]{0.33\linewidth}

\raisebox{-\height}{

\includegraphics[width=1\linewidth,height=\textheight,keepaspectratio,alt={Linearithmic specification.}]{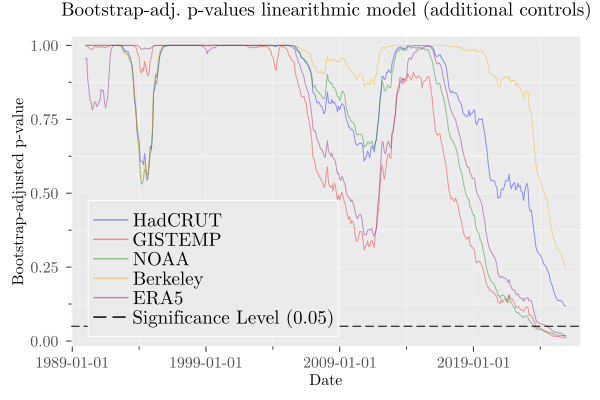}

}

\subcaption{\label{}Linearithmic specification.}
\end{minipage}%
\begin{minipage}[t]{0.33\linewidth}

\raisebox{-\height}{

\includegraphics[width=1\linewidth,height=\textheight,keepaspectratio,alt={Quadratic specification.}]{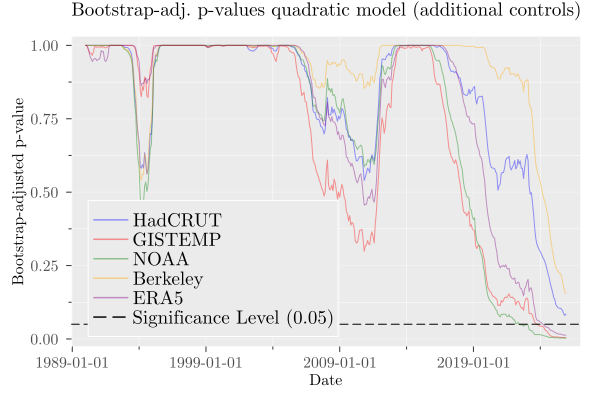}

}

\subcaption{\label{}Quadratic specification.}
\end{minipage}%
\begin{minipage}[t]{0.33\linewidth}

\raisebox{-\height}{

\includegraphics[width=1\linewidth,height=\textheight,keepaspectratio,alt={Free-curvature specification.}]{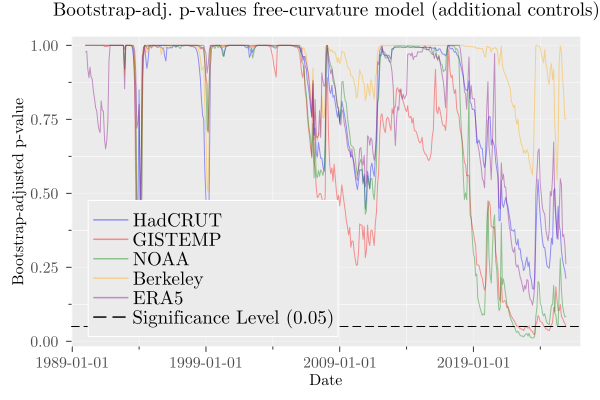}

}

\subcaption{\label{}Free-curvature specification.}
\end{minipage}%

\caption{\label{fig-bootstrap-pvals}Bootstrap-adjusted \(p\)-values for
the null of no acceleration, \(H_0: \beta = 0\), across expanding
windows ending between 1990 and 2025, for the three specifications
estimated with the El Niño and solar controls. The null tested is that
of a constant warming rate throughout, and in the free-curvature case
the test is on \(\beta\) alone, the exponent \(\hat\delta\) being a
shape estimate on which no inference is conducted. The dashed line marks
the 5\% level.}

\end{figure}%

Figure~\ref{fig-bootstrap-pvals} shows that rejections appear only in
the most recent windows. For the linearithmic specification the first
family-wise rejection occurs in the window ending 2023-09 for GISTEMP,
2023-07 for NOAA and 2024-07 for ERA5, while HadCRUT and Berkeley do not
reject in any window, reaching adjusted \(p\)-values of \(0.119\) and
\(0.246\) in the full sample. The quadratic specification behaves
similarly, first rejecting in 2023-10 for GISTEMP, 2022-06 for NOAA and
2024-04 for ERA5, and likewise never for HadCRUT and Berkeley.

These dates are not an artefact of the particular exponent that the
fixed-shape specifications impose. Fixing \(\delta\) on a grid from
\(1.25\) to \(3\) and repeating the procedure leaves the set of
rejecting datasets unchanged: GISTEMP, NOAA and ERA5 at every value,
HadCRUT and Berkeley at none; and moves the first rejection by at most
about a year. Among the exponents at or below the quadratic, the
linearithmic returns the largest final adjusted \(p\)-values in four of
the five datasets, though the differences do not change which datasets
reject. The SM reports the grid in full.

The gap between the raw and the adjusted evidence is what makes the
correction matter. Judged by unadjusted HAC \(p\)-values, the
linearithmic coefficient in the specification with controls is
significant at the \(5\%\) level from 2016-03 onwards for GISTEMP,
2017-05 for NOAA, 2017-12 for ERA5, 2019-12 for HadCRUT and 2023-07 for
Berkeley. Reading those dates as the onset of detectable acceleration
would overstate the evidence, since they are crossings of a threshold
tested repeatedly across the \(432\) nested windows. Once that repeated
testing is accounted for, the three datasets that still reject do so six
to seven and a half years later, and the remaining two do not reject in
any window.

The dates uncovered by the expanding-window design are dates of
detection rather than of onset, and on that reading they sit comfortably
alongside the existing literature. A change beginning in 2013 or 2014,
as \textcite{Forster2026} date it, would not be expected to clear a
family-wise threshold until enough subsequent record had accumulated to
make it visible, which is what the rejections from mid-2023 onwards
represent; the same lag accounts for the failure to establish a second
change point on data ending in 2023 \autocite{beaulieu_recent_2024}.
That agreement is informative because the specifications differ: the
earlier datings impose either a discrete break or a constant rate of
increase in the warming rate, while the test used here imposes neither.
The conclusion therefore does not rest on the functional form that
produced it.

The free-curvature specification,
Equation~\ref{eq-free-curvature-exogenous}, which fixes no shape, points
the same way but less sharply. GISTEMP rejects in \(22\) windows and
NOAA in \(20\), first in the windows ending 2022-04 and 2022-01
respectively; ERA5, HadCRUT and Berkeley never reject. In the full
sample, GISTEMP's adjusted \(p\)-value is \(0.050\), exactly at the
threshold, and NOAA's is \(0.084\). The rejections concentrate in
windows ending between early 2022 and mid-2023 and become intermittent
thereafter. The estimates admit two readings of this intermittency. As
observations from late 2023 onwards are added, both curvature parameters
grow: between the windows ending 2023-06 and 2025-12, NOAA's
\(\hat\delta\) rises from \(3.18\) to \(4.08\) and \(\hat\beta\) from
\(0.62\) to \(1.06\), and GISTEMP's from \(2.06\) to \(4.00\) and from
\(0.40\) to \(1.00\). Because the standard error on \(\hat\beta\) grows
faster, the intermittency may reflect precision lost to the free shape;
it may equally reflect weaker evidence once the shape is relaxed. The SM
shows that the estimates cannot distinguish the two. The agreement with
the fixed-shape specifications is accordingly partial: it holds on the
timing, but not on how many datasets reject.

The estimated exponent lies above the quadratic value in every dataset,
although we conduct no inference on it. Read as a shape estimate, it
places the curvature at or above the quadratic rather than at the
linearithmic value, and gives no sign that the choice of \(t\log(t)\)
manufactures curvature the data do not contain. The SM reports the
exponent path and the final-window estimates in full.

\subsection{What the specifications imply beyond the
sample}\label{sec-beyond-sample}

The specifications agree in sample, but they differ outside it by
construction. Extending each fitted trend to 2050 makes the size of that
difference explicit. These are extrapolations of an estimated
deterministic trend, not forecasts of temperature. Crucially, they carry
no information about emissions, forcings or feedbacks, but they show
what each functional form commits its user to; see the SM for the
construction and for the full results.

Figure~\ref{fig-projections} gives the comparison. At the end of the
sample the linearithmic specification puts the warming rate at \(0.251\)
to \(0.281\) °C decade\(^{-1}\) across the five datasets, against
\(0.196\) to \(0.216\) for the linear specification, whose rate is
constant by construction. By 2050 the linearithmic rate rises to
\(0.271\) to \(0.310\) and the quadratic to \(0.351\) to \(0.420\), a
gap of about \(0.09\) °C decade\(^{-1}\), and the implied anomalies
differ by roughly \(0.2\) °C. The free-curvature specification projects
a faster rise still.

\begin{figure}

\centering{

\includegraphics[width=1\linewidth,height=\textheight,keepaspectratio]{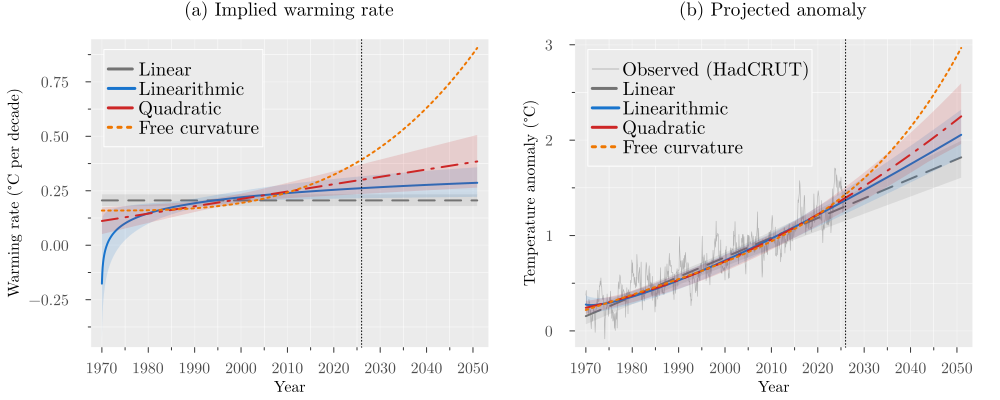}

}

\caption{\label{fig-projections}Projections under alternative trend
specifications, estimated on 1970--2025 with El Niño and solar controls
and extended to 2050. Panel (a) shows the implied warming rate, and
panel (b) the projected anomaly level with the observed HadCRUT series
behind it. Lines are means across the five temperature datasets; shaded
regions span the pointwise \(95\%\) HAC confidence bands of all five, so
they cover both estimation uncertainty and the disagreement between
datasets. The vertical dotted line marks the end of the estimation
sample. The free-curvature specification is drawn without a band, since
no inference is conducted on \(\hat\delta\), and excludes Berkeley,
whose exponent is pinned at the upper constraint; see the SM.}

\end{figure}%

The specifications therefore separate in the expected order. They do
not, however, separate statistically: the pointwise \(95\%\) HAC bands
for the linearithmic and quadratic projections overlap at every date to
2050 in all five datasets, for the rate and for the level alike. This
bounds what the present evidence can support. It separates a constant
warming rate from an increasing one, which is the hypothesis under test;
it does not identify which supralinear form the warming rate follows,
and the agreement between the specifications about the presence and
timing of curvature is the firm result.

That ordering is, however, the reason for leading with the linearithmic
specification. Of the accelerating forms considered it implies the
least: the lowest warming rate from 2025 onwards, the lowest projected
level in 2050, \(1.96\) to \(2.16\) °C, against \(2.18\) to \(2.38\) for
the quadratic and \(2.50\) to \(3.48\) where the curvature is left free,
and the only warming rate whose own growth slows as time passes. One
consequence follows: the warming it projects should be read as the
lowest of those considered rather than as a central estimate.

\section{Discussion and implications}\label{sec-discussion}

We set out to answer two questions: since when is acceleration in GMST
detectable, and does the answer depend on the specification used to look
for it? On the second, the answer is that it does not depend on it,
provided the specification tests curvature rather than a break. Over the
full 1970--2025 sample the linearithmic coefficient is positive and
significant in all five datasets once the El Niño effect and solar
variation are controlled for. The controls matter: without them the same
coefficient is significant in only three of the five datasets. The
quadratic specification, estimated on the same samples, gives closely
comparable answers throughout, and the free-curvature specification
agrees on direction, despite the three imposing very different amounts
of structure on the shape of the curvature.

On the first question, detection is recent. Under the expanding-window
design with a bootstrap correction for multiple testing, the evidence
appears only in the most recent windows and in three of the five
datasets. The uncorrected crossings arrive six to seven and a half years
earlier, and treating those as the onset would mistake the repeated
testing of an accumulating record for new information. The
free-curvature specification, which fixes no shape at all, agrees on
direction but rejects far less widely: in the full sample only GISTEMP
reaches the threshold, at exactly \(0.050\), with NOAA at \(0.084\) and
the remaining three well short. Its rejections also fall earlier,
concentrated in windows ending between early 2022 and mid-2023 and
intermittent thereafter, which may reflect precision lost to the free
shape or weaker evidence once the shape is relaxed.

Taken together, the two answers reconcile the conflicting findings in
the literature that motivated the paper. That
\textcite{beaulieu_recent_2024} found no detectable surge with data to
2023, that the break-based literature has established no second change
point since the one placed in the 1970s
\autocite{gadea-rivas_trends_2024,estrada_2013}, and that
\textcite{Forster2026} reported acceleration with over 98\% confidence
once the record ran to 2024 need not be read as three answers to the
same question. Detection turns on three things that differ across those
studies. It turns on where the sample ends: the earliest family-wise
rejection we find is in the window ending 2022-01, for NOAA under the
free-curvature specification, the linearithmic specification first
rejects in mid-2023, and ERA5 does not reject until 2024-07. An analysis
with data to 2023 was therefore looking at a record in which the
evidence had only begun to emerge, and in only some datasets. It turns
on the controls: with the El Niño and solar terms omitted, not one of
the \(2{,}160\) window-level tests rejects at the \(5\%\) level in
either OLS specification, the smallest adjusted \(p\)-value being
\(0.281\). The curvature is small relative to the interannual
variability it must be separated from, so a test that leaves that
variability in the residual will not find it at any sample end point
considered here. And it turns on the dataset: GISTEMP, NOAA and ERA5
reject while HadCRUT and Berkeley do not, in any window and under any
specification, so part of the disagreement in the literature is a
disagreement about which temperature product is analysed rather than
about the state of the climate.

Reconciling those findings is not the same as endorsing the strongest of
them. The evidence is weaker than a single fit at one sample end point
suggests, and the expanding-window design is what makes that visible:
a test applied once, at the end of the record, is scored as though it
had been available only then, when in fact the same test has been
available at every month through which the record has grown. What
survives the correction for that is a genuine but narrow result: three
of the five datasets, only in the last four years of a fifty-six-year
sample, and only once the additional variability is controlled for, with
the specification that fixes no shape at all reaching the threshold in
one dataset and missing it in the other four. Acceleration in GMST is
detectable, then, but with the confidence appropriate to a signal that has just
emerged rather than to an established feature of the record.

We read the agreement among the specifications as the finding, rather
than a contest between them. The linearithmic and quadratic trends
assume different shapes, a warming rate that rises ever more slowly
against one that rises at a constant pace, but on this sample their
curvature terms are correlated at \(0.97\) once the linear trend is
removed, so they are close to the same test, and the quadratic fits
better. That they reject a constant warming rate on the same data, at
similar dates, therefore says less about functional form than two
independent tests would. The check on functional form is the
free-curvature specification, which estimates the shape directly; while
too imprecise to reject as widely, it agrees with them on direction and
on timing.

Among the specifications considered, the linearithmic implies the
smallest departure from a constant warming rate, with an acceleration
that declines over time, and the free-curvature estimates place the data
above it. Extended to 2050 the same ordering persists, with the
linearithmic specification projecting the lowest warming rate and the
lowest anomaly of the three. If the linearithmic form misrepresents the
trend, then, the more likely direction of the error is that it
understates the acceleration.

As established in Theorem~\ref{thm-breaks-supralinear-trend}, standard
break-point tests cannot distinguish a discrete regime shift from
continuous supralinear curvature. The breaks already reported in the
GMST record can be consistent with the single continuous pattern
documented here, and need not be read as separate events. For modelling
and policy purposes, the distinction still matters: a piecewise linear
specification is fitted assuming discrete change, so it will continue to
underrepresent the trend in GMST if the true warming rate is in fact
continuously increasing, capturing only one-time shifts in the slope
rather than a persistent rise in curvature.

The presence of acceleration suggests that the rate of warming is
increasing over time, which could lead to more severe impacts on
ecosystems, human health, and the economy. It also highlights the
urgency of taking action to mitigate greenhouse gas emissions and adapt
to the changing climate. There is already evidence that the airborne
fraction of anthropogenic carbon emissions has been increasing in recent
decades \autocite{vera-valdes2026}, and that Earth's energy imbalance
has more than doubled over the same period
\autocite{mauritsenEarthsEnergyImbalance2025}, which is consistent with
the observed acceleration in GMST.

\subsection*{References}\label{references}
\addcontentsline{toc}{subsection}{References}

\printbibliography[heading=none]

\section*{Author statements}\label{author-statements}
\addcontentsline{toc}{section}{Author statements}

\subsection*{Data availability
statement}\label{data-availability-statement}
\addcontentsline{toc}{subsection}{Data availability statement}

Data for this study were obtained from publicly available sources.
Global temperature anomalies are obtained from five widely used
datasets: HadCRUT \autocite{HadCRUT5}, GISTEMP \autocite{GISTEMP}, NOAA
\autocite{NOAA}, Berkeley Earth \autocite{BerkeleyEarth}, and ERA5
\autocite{ERA5}, as anomaly series relative to the 1850--1900 mean. ERA5
begins in 1940, GISTEMP in 1880 and the remaining datasets in 1850.

ERA5 requires an additional step, as it is distributed as a gridded
reanalysis rather than as a global anomaly series. Global monthly means
are formed from the raw grid with latitude weights proportional to the
cosine of latitude, so that cells are weighted by area, and converted
from Kelvin to degrees Celsius. Since ERA5 supplies no fixed reference
period, we compute a monthly climatology over 1991--2020, take anomalies
relative to it, and add the constant offset of \(0.88\) °C given in the
ERA5 documentation \autocite{C3SReferencePeriods} to express them
relative to 1850--1900. Being additive, the offset shifts the intercept
only and leaves every trend coefficient and test unchanged.

Data for the ONI are obtained from the National Oceanic and Atmospheric
Administration (NOAA) \autocite{SST_Part1,SST_Part2}. Solar variation is
measured by the sunspot number, a common proxy for solar irradiance
\autocite{Forster2026}, obtained from WDC-SILSO, Royal Observatory of
Belgium \autocite{SILSO}. All series are used at monthly frequency,
which captures dynamics in GMST and in the ENSO and solar controls that
annual averaging would obscure.

The harmonised data are retrieved from
\url{https://everval.github.io/Global-Temperature-Anomalies/}
\autocite{vera-valdesGlobalTemperatureAnomalies}.

\subsection*{Reproducibility statement}\label{reproducibility-statement}
\addcontentsline{toc}{subsection}{Reproducibility statement}

The code used for data processing, analysis, and visualisation in this
study is archived on Zenodo at DOI:
\href{https://zenodo.org/records/23323750}{10.5281/zenodo.23323749},
together with the estimation output of every specification fitted in the
course of the analysis. The archived version is a release of the
manuscript's GitHub repository,
\href{https://github.com/everval/Supralinear-Warming-Acceleration}{github.com/everval/Supralinear-Warming-Acceleration}.
All analyses were conducted using Julia version 1.12.6.
\texttt{Project.toml} and \texttt{Manifest.toml} files are included in
the repository to ensure that the computational environment can be
replicated.

\subsection*{Competing interests}\label{competing-interests}
\addcontentsline{toc}{subsection}{Competing interests}

The author declares no competing interests.

\subsection*{Funding}\label{funding}
\addcontentsline{toc}{subsection}{Funding}

The author received no specific funding for this work.

\subsection*{Acknowledgements}\label{acknowledgements}
\addcontentsline{toc}{subsection}{Acknowledgements}

The author thanks the editor and two anonymous reviewers for their
constructive comments, which greatly improved the manuscript. The author
also thanks participants of the 2025 Conference on Econometric Models of
Climate Change at the University of Victoria, Canada, and of the
Econometrics Seminar at Aarhus University, Denmark, for their helpful
comments. Particular thanks are due to Mikkel Bennedsen for his feedback
on an earlier draft.

\newpage{}

\section*{Supplementary Material}\label{supplementary-material}
\addcontentsline{toc}{section}{Supplementary Material}

\subsection*{Theoretical background}\label{theoretical-background}
\addcontentsline{toc}{subsection}{Theoretical background}

Acceleration is defined as a positive second derivative of the trend
function with respect to time.

\begin{definition}[Acceleration]\protect\hypertarget{def-acceleration}{}\label{def-acceleration}

Let \(x_t\) be the trend in a time series. We say that there is
acceleration if:
\begin{equation}\protect\phantomsection\label{eq-acceleration}{
\frac{d^2 x_t}{dt^2} > 0.
}\end{equation}

\end{definition}

Note that a linear trend in GMST would mean no acceleration. Hence, the
existence of acceleration in a time series is related to the concept of
it having a supralinear trend, or one that grows faster than a linear
trend.

Consider a general polynomial trend specification of the form:

\begin{equation}\protect\phantomsection\label{eq-polynomial-trend}{
x_t = \alpha + \lambda t + \beta t^\delta,
}\end{equation}

where \(\alpha\) is the intercept, \(\lambda\) is the slope of the
linear trend, \(\beta\) is the coefficient associated with the
polynomial term, and \(\delta\) is its degree. When \(\delta > 1\) it
captures supralinear trends.

Lemma~\ref{lem-linearithmic-speed} shows that the linearithmic trend
grows faster than a linear trend but more slowly than any polynomial
trend with degree greater than one.

\begin{lemma}[]\protect\hypertarget{lem-linearithmic-speed}{}\label{lem-linearithmic-speed}

Let \(x_t\) follow a linearithmic trend as in
Equation~\ref{eq-linearithmic-trend} with \(\beta > 0\). Then, the speed
of growth of \(x_t\) satisfies:

\begin{itemize}
\item
  The speed is faster than linear: \[
  \lim_{t \to \infty} \frac{x_t}{t} = \infty.
  \]
\item
  The speed is slower than polynomial for any \(\delta > 1\): \[
  \lim_{t \to \infty} \frac{x_t}{t^\delta} = 0.
  \]
\end{itemize}

\end{lemma}

\begin{proof}
For the first part, note that:

\[
\lim_{t \to \infty} \frac{x_t}{t} = \lim_{t \to \infty} \frac{\alpha + \lambda t + \beta \log(t) t}{t} = \lim_{t \to \infty} \left[ \frac{\alpha}{t} + \lambda + \beta \log(t) \right] = \infty,
\]

since the first term vanishes, the second is constant, and the third
diverges for \(\beta > 0\).

For the second part, note that for \(\delta > 1\) the terms
\(\alpha / t^\delta\) and \(\lambda t^{1-\delta}\) both vanish, so that:
\[
\lim_{t \to \infty} \frac{x_t}{t^\delta} = \lim_{t \to \infty} \frac{\beta \log(t) t}{t^\delta} = \lim_{t \to \infty} \frac{\beta \log(t)}{t^{\delta - 1}},
\]

which is one of the indeterminate forms covered by L'Hôpital's rule.
Applying the rule, we get:

\[
\lim_{t \to \infty} \frac{\beta \log(t)}{t^{\delta - 1}} = \lim_{t \to \infty} \frac{\beta / t}{(\delta - 1) t^{\delta - 2}} = \lim_{t \to \infty} \frac{\beta}{(\delta - 1) t^{\delta - 1}} = 0,
\]

which concludes the proof.
\end{proof}

Hence, the linearithmic trend grows more slowly than every supralinear
trend in the polynomial specification. It is not a lower bound for all
supralinear trends, however: \(t \log(\log(t))\), for instance, grows
faster than a linear trend but more slowly than the linearithmic one.

The next theorem presents the conditions under which the linearithmic,
polynomial and free-curvature specifications capture acceleration in the
sense of a positive second derivative of the trend function.

\begin{theorem}[]\protect\hypertarget{thm-supralinear-trend}{}\label{thm-supralinear-trend}

Let \(x_t\) be a function of time. \(x_t\) is \emph{accelerating} if any
of the following conditions hold:

\begin{enumerate}
\def\labelenumi{\arabic{enumi}.}
\item
  \(\beta\) is positive in the linearithmic specification,
  Equation~\ref{eq-linearithmic-trend}.
\item
  \(\delta > 1\) and \(\beta>0\) in the polynomial specification,
  Equation~\ref{eq-polynomial-trend}.
\item
  \(\delta > 0\) and \(\beta > 0\) in the free-curvature specification,
  Equation~\ref{eq-free-curvature}.
\end{enumerate}

\end{theorem}

\begin{proof}
The proof follows directly from the definition of acceleration,
Equation~\ref{eq-acceleration}.

For 1., note that: \[
\frac{d^2 x_t}{dt^2} = \frac{d^2 (\alpha + \lambda t)}{dt^2} + \beta \frac{d^2 \log(t)t}{dt^2} = \beta \frac{1}{t} > 0,
\]

if \(\beta > 0\) and \(t > 0\).

For 2., note that:

\[
\frac{d^2 x_t}{dt^2} = \frac{d^2 (\alpha + \lambda t)}{dt^2} + \beta \frac{d^2 (t^\delta)}{dt^2} = \beta \delta (\delta - 1) t^{\delta - 2} > 0,
\]

if \(\beta > 0\), \(\delta > 1\) and \(t > 0\).

For 3., Equation~\ref{eq-free-curvature} has second derivative:

\[
\frac{d^2  \beta g(t,\delta)}{dt^2} =  \beta \frac{\delta(\delta-1)t^{\delta-2}}{\delta-1} =  \beta \delta\, t^{\delta-2} > 0,
\]

if \(\beta > 0\), \(\delta > 0\) and \(t > 0\). The expression is
continuous at \(\delta = 1\), where it returns \(\beta/t\); condition 1
is therefore the special case \(\delta = 1\) of condition 3.
\end{proof}

Condition 3 differs from condition 2 in a way that matters for the test.
In the polynomial specification the second derivative carries the factor
\(\delta - 1\), whose sign changes at \(\delta = 1\), so acceleration
requires a restriction on \(\delta\) as well as on \(\beta\). In the
free-curvature specification the second derivative carries the factor
\(\delta\) instead, and the sign of the second derivative is then the
sign of \(\beta\) alone. Within this parametrization \(\delta\) governs
how the curvature is distributed over the sample, not whether it is
positive, so \(H_0: \beta = 0\) is the hypothesis of no acceleration
whatever the value of \(\delta\). This is why the inference reported
below rests on \(\hat\beta\), with \(\hat\delta\) treated as a shape
estimate.

On the other hand, as \(\delta \to 0\) the curvature term tends to
\(t - 1\), an exact linear function of time, so in that limit \(\beta\)
is not separately identified from the coefficient on the linear trend.
The estimation therefore restricts \(\delta \geq 1\), which is both the
alternative of interest, curvature at least linearithmic, and well away
from the linear boundary.

Furthermore, the next theorem presents the theoretical design used to
test for acceleration in GMST. We write \(X_T = O_p(a_T)\) if
\(X_T/a_T\) is bounded in probability as \(T \to \infty\).

\begin{theorem}[]\protect\hypertarget{thm-testing-acceleration}{}\label{thm-testing-acceleration}

Let \(C_t = h_t + \varepsilon_t\), where \(\varepsilon_t\) is \(I(0)\), for \(t = 1,\ldots, T\); with
\(T\in \mathbb{N}\) and \(h_t\) the trend component of \(C_t\). Consider
the linearithmic specification, estimated using OLS jointly with a
linear trend:
\begin{equation}\protect\phantomsection\label{eq-trend-specification-appendix}{
C_t = \alpha + \lambda t + \beta \log(t)t + u_t.
}\end{equation}

Then, the following holds:

\begin{itemize}
\tightlist
\item
  If \(h_t\) follows a linear trend, \(\hat{\beta} = O_p(T^{-3/2})\) so
  that it goes to zero as \(T\) increases.
\item
  If \(h_t\) follows a linearithmic trend with coefficient
  \(\beta_0 > 0\), \(\hat{\beta} \to_p \beta_0\), a positive constant.
\item
  If \(h_t\) follows a polynomial trend with degree \(\delta > 1\),
  \(\hat{\beta}/T^{\delta - 1}\) converges in probability to a positive
  constant, so \(\hat{\beta} \to \infty\) as \(T\) increases.
\end{itemize}

\end{theorem}

\begin{proof}
Write Equation~\ref{eq-trend-specification-appendix} as
\(C_t = w_t' \theta + u_t\) with \(w_t = (1, t, z_t)'\) and
\(z_t = \log(t) t\), and let: \[
m_t := z_t - \pi_0 - \pi_1 t, \qquad (\pi_0, \pi_1) := \arg\min_{a,b} \sum_{t=1}^T (z_t - a - b t)^2,
\] denote the residual from projecting the linearithmic regressor on a
constant and a linear trend. By the Frisch-Waugh-Lovell theorem
\autocite{FrischWaugh1933,Lovell1963} the OLS estimator of \(\beta\) in
Equation~\ref{eq-trend-specification-appendix} is given by:
\begin{equation}\protect\phantomsection\label{eq-fwl-beta}{
\hat{\beta} = \frac{\sum_{t=1}^T m_t C_t}{\sum_{t=1}^T m_t^{2}}.
}\end{equation}

The key quantity is therefore the partialled-out variation in \(z_t\),
not its raw variation. Since \(z_t = t\log(T) + t\log(t/T)\), the part
of \(z_t\) carrying \(\log(T)\) is an exact multiple of \(t\) and is
annihilated by the projection. Writing the residual sum of squares of
that projection in the usual form: \[
\sum_{t=1}^T m_t^{2} = \left( \sum_{t=1}^T z_t^2 - T\bar{z}^2 \right) - \frac{\left( \sum_{t=1}^T t z_t - T \bar{t}\bar{z} \right)^2}{\sum_{t=1}^T t^2 - T \bar{t}^2},
\]

with \(\bar{t}\) and \(\bar{z}\) the sample means of \(t\) and \(z_t\),
and substituting the elementary sums \autocite{deBruijn1958Asymptotic}:
\[
\sum_{t=1}^T t^2 = \frac{T^3}{3} + O(T^2), \qquad \sum_{t=1}^T t\log t = \frac{T^2}{2}\log T - \frac{T^2}{4} + O(T\log T),
\] \[
\sum_{t=1}^T t^2\log t = \frac{T^3}{3}\log T - \frac{T^3}{9} + O(T^2\log T), \qquad \sum_{t=1}^T t^2\log^2 t = \frac{T^3}{3}\log^2 T - \frac{2T^3}{9}\log T + \frac{2T^3}{27} + O(T^2\log^2 T),
\] gives:
\begin{equation}\protect\phantomsection\label{eq-partialled-denominator}{
\sum_{t=1}^T m_t^{2} = T^3\left[ \frac{\log^2 T}{12} + \frac{\log T}{36} + \frac{5}{432} \right] - T^3\left[ \frac{\log^2 T}{12} + \frac{\log T}{36} + \frac{1}{432} \right] + o(T^3) = \frac{T^3}{108} + o(T^3).
}\end{equation}

The first bracket is the total variation in \(z_t\) and the second the
part of it explained by a linear trend, and the \(\log^2 T\) and
\(\log T\) terms cancel exactly between them: those are the terms a
linear trend reproduces. What survives is of order \(T^3\).

Decompose the numerator of Equation~\ref{eq-fwl-beta} as
\(\sum_t m_t h_t + \sum_t m_t \varepsilon_t\). Since \(m_t = O(T)\)
uniformly and \(\varepsilon_t\) is \(I(0)\), the stochastic term
satisfies
\(\sum_t m_t \varepsilon_t = O_p\big((\sum_t m_t^2)^{1/2}\big) = O_p(T^{3/2})\).
For the deterministic term, the projection annihilates any exactly
linear function of \(t\), so: \[
\sum_{t=1}^T m_t h_t =
  \begin{cases}
    0, & \text{if } h_t = c\,t\\
    \beta_0 \, T^3/108 + o(T^3), & \text{if } h_t = \beta_0 \log(t) t\\
    \kappa_\delta T^{\delta + 2} + o(T^{\delta + 2}), & \text{if } h_t = t^\delta, \delta > 1,
  \end{cases}
\] where \(\kappa_\delta > 0\) is the limit of
\(T^{-\delta-2}\sum_t m_t t^\delta\), positive because \(m_t\) and
\(t^\delta\) are both convex in \(t\) once the linear trend is
partialled out.

The three cases of the theorem follow on dividing by
Equation~\ref{eq-partialled-denominator}. Under a linear trend the
deterministic term vanishes and only the stochastic term remains, so
\(\hat{\beta} = O_p(T^{3/2}/T^3) = O_p(T^{-3/2})\). Under a linearithmic
trend the deterministic term dominates and
\(\hat{\beta} \to_p \beta_0\). Under a polynomial trend of degree
\(\delta > 1\), the deterministic term dominates and
\(\hat{\beta} = 108\,\kappa_\delta T^{\delta - 1}(1 + o_p(1))\), which
diverges.
\end{proof}

Theorem~\ref{thm-testing-acceleration} shows that the conditions for
acceleration can be tested by evaluating the behaviour of the OLS
estimator \(\hat{\beta}\) in the linearithmic specification. Two
features of the result are worth drawing out.

First, under a linear trend \(\hat{\beta}\) does not merely shrink, it
is centred at zero and vanishes at the super-consistent rate
\(T^{-3/2}\). This is what allows the use of a conventional two-sided
test of \(H_0: \beta = 0\) as a test of a constant warming rate: the
usual normal approximation to the distribution of the \(t\)-statistic
applies asymptotically. In finite samples with persistent errors that
approximation is poor, which is why the size of the HAC test is
corrected in Section 3 and the SM.

Second, under a polynomial alternative with \(\delta > 1\),
\(\hat{\beta}\) diverges at the rate \(T^{\delta-1}\), which is
increasing in \(\delta\). The test is therefore consistent against the
entire supralinear polynomial class.

Note that the linearithmic specification must be understood as an
approximation to an unknown trend function, potentially stochastic, of
temperature. The true trending dynamics may be more complex, but the
specification provides a useful framework for testing for acceleration
in temperature, since it grows more slowly than any supralinear
polynomial trend. Hence, if the conditions for acceleration in
Equation~\ref{eq-linearithmic-trend} are satisfied, then there is
evidence that the trending dynamics are increasing at a supralinear
rate. Furthermore, the trend does not have to be deterministic; the
analysis extends to stochastic trends growing at supralinear rates,
whose order can be characterised through summability
\autocite{berenguer2014summability}.

\subsubsection*{Breaks or supralinear
trends?}\label{breaks-or-supralinear-trends}
\addcontentsline{toc}{subsubsection}{Breaks or supralinear trends?}

Classical tests for structural breaks assume that the trend is linear
within each segment determined by the breaks
\autocite{andrews1993tests}. Formally, consider the following segmented
linear trend model:

\[
C_t = \alpha_1 + \lambda_1 t + u_t, \quad t = 1, 2, \ldots, k,
\] \[
C_t = \alpha_2 + \lambda_2 t + u_t, \quad t = k+1, \ldots, T,
\]

where \(k\) is the break point, \(T\) is the total number of
observations, and \(u_t\) is a stationary error term. Each segment has
its own intercept and slope, so a break may shift the level of the
trend, its slope, or both.

Define the sum of squared residuals of the break model, fitted by OLS on
each segment separately, as a function of the break point \(k\):

\[
S_T(k) := \min_{a_1, b_1} \sum_{t=1}^k (C_t - a_1 - b_1 t)^2 + \min_{a_2, b_2} \sum_{t=k+1}^T (C_t - a_2 - b_2 t)^2.
\]

We compare it against the no-break benchmark; that is, the sum of
squared residuals of a single linear trend fitted over the whole sample:

\[
S_T^0 := \min_{a, b} \sum_{t=1}^T (C_t - a - b t)^2.
\]

Tests for structural breaks estimate the break point \(\hat{k}\) as the
value of \(k\) that minimizes the sum of squared residuals across both
segments, searching over a trimmed set of candidate break dates
\autocite{andrews1993tests}: \[
\hat{k} = \text{argmin}_{k \in \mathcal{K}_\tau} S_T(k), \quad \mathcal{K}_\tau := \{k : \lfloor \tau T \rfloor \leq k \leq \lfloor (1-\tau) T \rfloor \},
\]

for a trimming parameter \(\tau \in (0, 1/2)\). Trimming excludes the
degenerate endpoints, at which one of the two segments contains too few
observations to identify its slope.

Because the break model nests the single linear trend, which it
reproduces by setting the coefficients of the two segments equal,
splitting the sample can never increase the sum of squared residuals:
\(S_T(k) \leq S_T^0\) for every \(k\), whatever the data. A reduction is
therefore not in itself evidence of a break; what matters is its size
relative to \(S_T^0\). In its basic form, the test statistic of
\textcite{andrews1993tests} is the supremum over
\(k \in \mathcal{K}_\tau\) of the \(F\) statistic for the null of no
break, \[
F_T(k) := \frac{\big(S_T^0 - S_T(k)\big)/2}{S_T(k)/(T-4)} = \frac{T-4}{2}\, \frac{\rho_T(k)}{1 - \rho_T(k)}, \qquad \rho_T(k) := \frac{S_T^0 - S_T(k)}{S_T^0},
\] which is increasing in the relative reduction \(\rho_T(k)\).

Theorem~\ref{thm-breaks-supralinear-trend} shows that, if the true trend
is linearithmic or polynomial with degree greater than one, the relative
reduction stays bounded away from zero, uniformly over the candidate
break dates, so that the test for structural breaks rejects with
probability approaching one. In contrast,
Lemma~\ref{lem-breaks-linear-trend} shows that the relative reduction
vanishes if the true trend is linear.

\begin{theorem}[]\protect\hypertarget{thm-breaks-supralinear-trend}{}\label{thm-breaks-supralinear-trend}

Let \(C_t = h_t + \varepsilon_t\), where \(\varepsilon_t\) is \(I(0)\)
and \(h_t\) is either a linearithmic trend as in
Equation~\ref{eq-linearithmic-trend} or a polynomial trend as in
Equation~\ref{eq-polynomial-trend} with \(\delta > 1\), with
\(\beta > 0\) in both cases. Then there is a constant \(\rho_\tau > 0\),
not depending on \(T\), such that, with probability approaching one as
\(T \to \infty\), \[
\frac{S_T^0 - S_T(k)}{S_T^0} \geq \rho_\tau \quad \text{for every } k \in \mathcal{K}_\tau,
\] and \(\sup_{k \in \mathcal{K}_\tau} F_T(k) \to_p \infty\).

\end{theorem}

\begin{proof}
First, note that the sums of squared residuals depend on the trend only
through its deviation from linearity. A function that is linear in \(t\)
is fitted exactly by the single trend and by each segment, so it changes
neither \(S_T^0\) nor \(S_T(k)\). Moreover, dividing the data by
\(\beta g_T\), where \(g_T := T\) for the linearithmic trend and
\(g_T := T^\delta\) for the polynomial trend, divides both \(S_T^0\) and
\(S_T(k)\) by \((\beta g_T)^2\), so that \(\rho_T(k)\) and \(F_T(k)\)
are unchanged.

Hence, using that \(t\log(t) = t\log(T) + t\log(t/T)\), we can drop the
linear part of \(h_t\) and standardise the data by \(\beta g_T\). The
trend becomes \(\tilde h_t := f(t/T)\), where \[
f(u) := u\log(u) \quad \text{and} \quad f(u) := u^\delta,
\] for the linearithmic and the polynomial trend, respectively, and the
error becomes \(\varepsilon_t/(\beta g_T)\). That is, both trends are
written as a fixed function evaluated at \(t/T\), so that the argument
below covers both cases. From now on, \(S_T^0\) and \(S_T(k)\) refer to
the standardised data.

Let \(\tilde S_T^0\) and \(\tilde S_T(k)\) be the sums of squared
residuals computed from \(\tilde h_t = f(t/T)\) alone. We first show
that the relative reduction is bounded away from zero without the error,
and then that adding the error does not change this. In what follows,
\(c_1\) and \(c_2\) denote positive constants that do not depend on
\(T\) or \(k\).

To compare the two fits, treat
\(\tilde h := (\tilde h_1, \ldots, \tilde h_T)'\) as a vector in
\(\mathbb{R}^T\). Let \(V_0\) be the set of vectors that are linear in
\(t\) over the whole sample, and \(V_k\) the set of vectors that are
linear in \(t\) on each segment, \(t \leq k\) and \(t > k\), with
possibly different intercepts and slopes. A single line is linear on
each segment, taking the same coefficients on both, so that
\(V_0 \subset V_k\).

The OLS fits are the orthogonal projections of \(\tilde h\) onto these
subspaces, which we denote by \(P_0 \tilde h\) and \(P_k \tilde h\);
since the segments do not overlap, fitting each segment separately gives
the projection onto \(V_k\). Hence,
\(\tilde S_T^0 = \|\tilde h - P_0 \tilde h\|^2\) and
\(\tilde S_T(k) = \|\tilde h - P_k \tilde h\|^2\). Now write: \[
\tilde h - P_0 \tilde h = (\tilde h - P_k \tilde h) + (P_k \tilde h - P_0 \tilde h).
\] The first term is the residual of the segmented fit, which, by the
normal equations on each segment, is orthogonal to every vector in
\(V_k\). The second term is in \(V_k\), since \(P_k \tilde h \in V_k\)
and \(P_0 \tilde h \in V_0 \subset V_k\). By Pythagoras' theorem, it
follows that: \[
\tilde S_T^0 = \tilde S_T(k) + \|P_k \tilde h - P_0 \tilde h\|^2.
\]

Let \(a_0 + b_0 t\) be the single trend fitted to \(\tilde h_t\), and
let \(a_1 + b_1 t\) and \(a_2 + b_2 t\) be the trends fitted to the
first \(k\) and the last \(T-k\) observations, respectively. Writing out
the norm segment by segment gives: \[
\tilde S_T^0 - \tilde S_T(k) = \sum_{t=1}^k \big( a_1 - a_0 + (b_1 - b_0) t \big)^2 + \sum_{t=k+1}^T \big( a_2 - a_0 + (b_2 - b_0) t \big)^2.
\]

The intercepts in this expression can be dropped. For a line
\(\alpha + st\) over \(n\) consecutive values of \(t\), with mean
\(\bar t\), write \(\alpha + st = (\alpha + s\bar t) + s(t - \bar t)\).
Since \(\sum_t (t - \bar t) = 0\), the cross term vanishes, so that: \[
\sum_t (\alpha + st)^2 = n(\alpha + s\bar t)^2 + s^2 \sum_t (t - \bar t)^2 \geq s^2\, \frac{n(n^2-1)}{12},
\] where \(n(n^2-1)/12\) is the sum of squared deviations of \(n\)
consecutive integers from their mean.

Moreover, since \(b_2 - b_1 = (b_2 - b_0) - (b_1 - b_0)\), the triangle
inequality gives: \[
|b_2 - b_1| \leq |b_1 - b_0| + |b_2 - b_0| \leq 2\max\{|b_1 - b_0|, |b_2 - b_0|\},
\] so that the slope fitted on at least one segment, say segment \(j\),
differs from \(b_0\) by at least \(|b_2 - b_1|/2\). That is, the single
slope cannot be close to both segment slopes. Keeping only the sum over
segment \(j\), which has \(n_j \geq \lfloor \tau T \rfloor\)
observations, it follows that, for \(T\) large: \[
\tilde S_T^0 - \tilde S_T(k) \geq \frac{(b_2 - b_1)^2}{4}\, \frac{n_j(n_j^2 - 1)}{12} \geq c_1 T^3 (b_2 - b_1)^2.
\]

It is therefore sufficient to bound the difference between the two
fitted slopes from below. To do so, define the changes \[
d_t := \tilde h_t - \tilde h_{t-1}.
\] On a segment with \(n\) observations, the fitted slope is: \[
\frac{\sum_t (t - \bar t)\, \tilde h_t}{\sum_t (t - \bar t)^2},
\] where the sums run over the segment and \(\bar t\) is the mean of
\(t\) over it. Writing \(\tilde h_t\) as its first value in the segment
plus the changes up to \(t\), and exchanging the order of summation,
gives: \[
\sum_t (t - \bar t)\, \tilde h_t = \sum_r w_r d_r, \qquad w_r := \sum_{t \geq r} (t - \bar t),
\] where \(r\) runs over the segment except its first observation, and
the first value drops out since \(\sum_t (t - \bar t) = 0\). If \(r\) is
the \(i\)th observation of the segment, then \(w_r = (i-1)(n-i+1)/2\).
That is, the weights are positive and symmetric about the middle of the
segment, and they add up to \(\sum_t (t - \bar t)^2\). Hence, the fitted
slope is a weighted average of the changes within the segment.

As \(f\) is convex, \(d_t\) increases with \(t\). In the first segment
the changes are \(d_2, \ldots, d_k\), so that \(b_1 \leq d_k\). In the
second segment the changes are \(d_{k+2}, \ldots, d_T\), all of them at
least \(d_k\), and, by the symmetry of the weights, the changes from
\(m := \lfloor (k+T)/2 \rfloor + 1\) onwards carry at least half of the
weight, each of them at least \(d_m\). Thus, \(b_2 \geq (d_k + d_m)/2\),
so that: \[
b_2 - b_1 \geq \frac{d_m - d_k}{2}.
\]

By the mean value theorem, \(d_t = f'(\xi_t)/T\) for some
\(\xi_t \in ((t-1)/T, t/T)\), and then: \[
d_m - d_k = \frac{f'(\xi_m) - f'(\xi_k)}{T} = \frac{1}{T} \int_{\xi_k}^{\xi_m} f''(u)\, du \geq \frac{\xi_m - \xi_k}{T} \min_{u \in [\xi_k, \xi_m]} f''(u).
\] For \(k \in \mathcal{K}_\tau\) and \(T\) large,
\(\xi_k > (k-1)/T \geq \tau/2\) and, since
\(m - k \geq (T-k)/2 \geq \tau T/2\), also
\(\xi_m - \xi_k \geq \tau/4\). Moreover, since \(f(t/T)\) is the trend
with its linear part removed and divided by \(\beta g_T\), its second
derivative is that of the trend, rescaled: \[
f''(t/T) = \frac{T^2}{\beta g_T} \frac{d^2 h_t}{dt^2},
\] where \(d^2 h_t/dt^2\) is the second derivative computed in the proof
of Theorem~\ref{thm-supralinear-trend}. That is, \(f''\) is positive
because the trend is accelerating, and on \([\tau/2, 1]\) it is bounded
away from zero: \[
f''(u) = \frac{1}{u} \geq 1 \quad \text{and} \quad f''(u) = \delta(\delta - 1)u^{\delta-2} \geq \delta(\delta - 1) \min\{(\tau/2)^{\delta-2}, 1\},
\] for the linearithmic and the polynomial trend, respectively. Denoting
this lower bound by \(\kappa_\tau > 0\), it follows that: \[
b_2 - b_1 \geq \frac{\tau \kappa_\tau}{8T} =: \frac{c_2}{T}.
\]

Hence, \(\tilde S_T^0 - \tilde S_T(k) \geq c_1 c_2^2\, T\). On the other
hand, since \(|f(u)| \leq 1\) on \((0, 1]\): \[
\tilde S_T^0 \leq \sum_{t=1}^T \tilde h_t^2 \leq T.
\] Therefore, without the error, \[
\frac{\tilde S_T(k)}{\tilde S_T^0} \leq 1 - c_1 c_2^2, \quad \text{for every } k \in \mathcal{K}_\tau.
\]

Finally, consider the error. The residuals are obtained by applying a
projection to the data, and a projection cannot lengthen a vector. By
the triangle inequality, the error therefore changes the square roots of
both sums of squared residuals by at most \[
e_T := \frac{1}{\beta g_T} \left( \sum_{t=1}^T \varepsilon_t^2 \right)^{1/2} = O_p\big(T^{1/2}/g_T\big),
\] uniformly in \(k\), so that: \[
\left( \frac{S_T(k)}{S_T^0} \right)^{1/2} \leq \frac{\tilde S_T(k)^{1/2} + e_T}{(\tilde S_T^0)^{1/2} - e_T}.
\] Since \((\tilde S_T^0)^{1/2} \geq (c_1 c_2^2\, T)^{1/2}\) and
\(g_T \to \infty\), the ratio \(e_T/(\tilde S_T^0)^{1/2}\) is
\(O_p(g_T^{-1})\) and vanishes. Hence,
\(S_T(k)/S_T^0 \leq 1 - c_1 c_2^2 + o_p(1)\) uniformly in \(k\), and the
bound on the relative reduction holds with, for instance,
\(\rho_\tau := c_1 c_2^2/2\). The divergence of \(\sup_k F_T(k)\) then
follows from \(F_T(k) \geq \tfrac{T-4}{2}\, \rho_T(k)\), which concludes
the proof.
\end{proof}

\begin{lemma}[]\protect\hypertarget{lem-breaks-linear-trend}{}\label{lem-breaks-linear-trend}

Let \(C_t = h_t + \varepsilon_t\), where \(\varepsilon_t\) is \(I(0)\)
and \(h_t = \alpha + \lambda t\) is a linear trend. Then, as
\(T \to \infty\), \[
\frac{S_T^0 - S_T(k)}{S_T^0} = O_p(T^{-1}) \to 0,
\]

uniformly for \(k \in \mathcal{K}_\tau\). That is, splitting the sample
yields no asymptotic reduction in the sum of squared residuals relative
to the no-break benchmark, and \(\sup_{k \in \mathcal{K}_\tau} F_T(k)\)
is bounded in probability.

\end{lemma}

\begin{proof}
A linear trend is fitted exactly by the single trend and by each
segment, so \(h_t\) drops out and both sums of squared residuals depend
on the error alone. Their difference, \(S_T^0 - S_T(k)\), is then the
reduction in the sum of squared residuals of \(\varepsilon_t\) from
fitting two linear trends rather than one, which is \(O_p(1)\) uniformly
in \(k \in \mathcal{K}_\tau\) for an \(I(0)\) error, whereas
\(S_T^0 = T\sigma_\varepsilon^2\,(1 + o_p(1))\) with
\(\sigma_\varepsilon^2 := E(\varepsilon_t^2)\). The relative reduction
is therefore \(O_p(T^{-1})\) uniformly in \(k\), and
\(F_T(k) = \tfrac{T-4}{2}\, \rho_T(k)/(1 - \rho_T(k))\) is bounded in
probability.
\end{proof}

Our results suggest that the observed acceleration in global temperature
trends may be driven by supralinear trends rather than structural
breaks. In the short term, it may be difficult to distinguish between
these two phenomena. Nevertheless, if the trend remains supralinear, the
relative reduction from splitting the sample stays bounded away from
zero while the break statistic grows with the sample, so the null
hypothesis of no break in the linear trend will be rejected as the
sample size increases. This is not the case if the true trend is linear:
the reduction then vanishes, and the break statistic stays bounded
however long the sample.

\subsection*{Estimation details}\label{estimation-details}
\addcontentsline{toc}{subsection}{Estimation details}

\subsubsection*{Linearithmic and quadratic
specifications}\label{linearithmic-and-quadratic-specifications}
\addcontentsline{toc}{subsubsection}{Linearithmic and quadratic
specifications}

Coefficients are estimated by OLS, with HAC standard errors using a
Bartlett kernel and Newey-West bandwidth selection
\autocite{newey-west-1987,NeweyWest1994}. On the full sample the
selected bandwidth ranges from \(16.2\) to \(17.6\) months across
datasets for the linearithmic specification and from \(15.0\) to
\(16.6\) for the quadratic, so the estimator averages autocovariances
out to well beyond a year. To assess directly whether inference remains
valid given the serial correlation present, we simulate the size of the
test actually performed under a linear null with AR(1) errors of
persistence up to \(0.95\), and apply size-corrected critical values
simulated under ARMA(1,1) errors fitted to each dataset's residuals. Our
motivation for this approach is that it is more informative than
pre-testing the errors, and it avoids conditioning the reported
inference on a preliminary test whose own size is not controlled.

Including the linear trend costs precision. Over this sample
\(t\log(t)\) is correlated with \(t\) at \(0.9988\), against \(0.9683\)
for \(t^2\), giving variance inflation factors of \(405\) and \(16\).
The linear term therefore inflates the standard error on \(\hat\beta\)
by a factor of about \(20\) in the linearithmic specification and about
\(4\) in the quadratic. This affects the variance of \(\hat\beta\), not
its bias, and centring the time index does not reduce it, since centring
leaves the standard error on \(\hat\beta\) and its \(t\)-statistic
unchanged.

\subsubsection*{Model comparison}\label{sec-model-comparison}
\addcontentsline{toc}{subsubsection}{Model comparison}

Table~\ref{tbl-model-comparison} reports formal fit statistics for the
three OLS specifications on the full sample, each with the controls. All
three are estimated on identical data, so the criteria are directly
comparable.

\begin{longtable}[]{@{}lrrrrr@{}}
\caption{Fit statistics for the three OLS specifications, full sample
1970-01 to 2025-12 with El Niño and solar controls, \(T = 672\). Lower
AIC and BIC are better.}\label{tbl-model-comparison}\tabularnewline
\toprule\noalign{}
\textbf{Dataset / Model} & \(k\) & Adjusted \(R^2\) & AIC & BIC &
RMSE \\
\midrule\noalign{}
\endfirsthead
\toprule\noalign{}
\textbf{Dataset / Model} & \(k\) & Adjusted \(R^2\) & AIC & BIC &
RMSE \\
\midrule\noalign{}
\endhead
\bottomrule\noalign{}
\endlastfoot
\textbf{HadCRUT} & & & & & \\
Linear & 4 & 0.8902 & $-$2855.0 & $-$2836.9 & 0.1188 \\
Linearithmic & 5 & 0.8969 & $-$2896.3 & $-$2873.7 & 0.1150 \\
Quadratic & 5 & 0.8981 & $-$2904.5 & $-$2881.9 & 0.1143 \\
\textbf{GISTEMP} & & & & & \\
Linear & 4 & 0.8726 & $-$2755.9 & $-$2737.8 & 0.1279 \\
Linearithmic & 5 & 0.8850 & $-$2823.5 & $-$2801.0 & 0.1214 \\
Quadratic & 5 & 0.8870 & $-$2835.7 & $-$2813.2 & 0.1204 \\
\textbf{NOAA} & & & & & \\
Linear & 4 & 0.8889 & $-$2910.0 & $-$2892.0 & 0.1140 \\
Linearithmic & 5 & 0.9027 & $-$2998.2 & $-$2975.6 & 0.1066 \\
Quadratic & 5 & 0.9058 & $-$3020.4 & $-$2997.8 & 0.1049 \\
\textbf{Berkeley} & & & & & \\
Linear & 4 & 0.8879 & $-$2826.1 & $-$2808.0 & 0.1214 \\
Linearithmic & 5 & 0.8937 & $-$2860.8 & $-$2838.2 & 0.1181 \\
Quadratic & 5 & 0.8953 & $-$2871.3 & $-$2848.7 & 0.1172 \\
\textbf{ERA5} & & & & & \\
Linear & 4 & 0.8811 & $-$2727.4 & $-$2709.4 & 0.1306 \\
Linearithmic & 5 & 0.8930 & $-$2797.4 & $-$2774.9 & 0.1238 \\
Quadratic & 5 & 0.8950 & $-$2810.4 & $-$2787.8 & 0.1226 \\
\end{longtable}

Two comparisons matter, and they have very different magnitudes. Both
curvature specifications beat the linear model decisively: the BIC gap
from linear to linearithmic runs from \(30\) to \(84\) across the five
datasets. Between the two curvature specifications the gap is far
smaller and favours the quadratic, from \(8\) to \(22\). The quadratic
is thus the better-fitting model on this sample, on every criterion and
in every dataset, but the margin separating it from the linearithmic
form is a fifth to a third of the margin separating either from a
constant warming rate. On the usual scale both margins are strong
evidence: the data favour the quadratic shape over the linearithmic, and
favour either over a constant warming rate by much more.

\subsubsection*{Selection of the ENSO lag}\label{sec-enso-lag}
\addcontentsline{toc}{subsubsection}{Selection of the ENSO lag}

The El Niño control enters as \(N_{t-k}\) rather than \(N_t\). We select
\(k\) by goodness of fit over \(k = 0, \ldots, 8\) months. The lagged
series are constructed on the full monthly panel, which begins well
before 1970, so every candidate is estimated on an identical sample of
\(T = 672\) observations with an identical number of parameters;
\(R^2\), adjusted \(R^2\), AIC and BIC are therefore monotone
transformations of the residual sum of squares and rank the candidates
identically. For each specification we select the lag maximising the
mean \(R^2\) across the five datasets.

Table~\ref{tbl-enso-lag} reports the outcome. Every specification
selects \(k = 3\) months, unanimously across all five datasets,
improving mean \(R^2\) by roughly \(0.017\) over the contemporaneous
control.

\begin{longtable}[]{@{}
  >{\raggedright\arraybackslash}p{(\linewidth - 10\tabcolsep) * \real{0.3182}}
  >{\raggedright\arraybackslash}p{(\linewidth - 10\tabcolsep) * \real{0.1364}}
  >{\raggedleft\arraybackslash}p{(\linewidth - 10\tabcolsep) * \real{0.1364}}
  >{\raggedleft\arraybackslash}p{(\linewidth - 10\tabcolsep) * \real{0.1364}}
  >{\raggedleft\arraybackslash}p{(\linewidth - 10\tabcolsep) * \real{0.1364}}
  >{\raggedleft\arraybackslash}p{(\linewidth - 10\tabcolsep) * \real{0.1364}}@{}}
\caption{Selection of the ENSO lag \(k\) by mean \(R^2\) across the five
temperature datasets, over \(k = 0, \ldots, 8\) months. Sample 1970-01
to 2025-12, \(T = 672\) for every candidate lag. Adjusted \(R^2\), AIC
and BIC select the same lag in every case, and the selected lag is the
same for all five datasets in every
specification.}\label{tbl-enso-lag}\tabularnewline
\toprule\noalign{}
\begin{minipage}[b]{\linewidth}\raggedright
\textbf{Specification}
\end{minipage} & \begin{minipage}[b]{\linewidth}\raggedright
Estimator
\end{minipage} & \begin{minipage}[b]{\linewidth}\raggedleft
Selected \(k\)
\end{minipage} & \begin{minipage}[b]{\linewidth}\raggedleft
Mean \(R^2\) at \(k\)
\end{minipage} & \begin{minipage}[b]{\linewidth}\raggedleft
Mean \(R^2\) at \(k=0\)
\end{minipage} & \begin{minipage}[b]{\linewidth}\raggedleft
Gain
\end{minipage} \\
\midrule\noalign{}
\endfirsthead
\toprule\noalign{}
\begin{minipage}[b]{\linewidth}\raggedright
\textbf{Specification}
\end{minipage} & \begin{minipage}[b]{\linewidth}\raggedright
Estimator
\end{minipage} & \begin{minipage}[b]{\linewidth}\raggedleft
Selected \(k\)
\end{minipage} & \begin{minipage}[b]{\linewidth}\raggedleft
Mean \(R^2\) at \(k\)
\end{minipage} & \begin{minipage}[b]{\linewidth}\raggedleft
Mean \(R^2\) at \(k=0\)
\end{minipage} & \begin{minipage}[b]{\linewidth}\raggedleft
Gain
\end{minipage} \\
\midrule\noalign{}
\endhead
\bottomrule\noalign{}
\endlastfoot
Linearithmic, Equation~\ref{eq-linearithmic-trend-exogenous} & OLS & 3 &
0.8949 & 0.8779 & 0.0170 \\
Quadratic, Equation~\ref{eq-quadratic-trend-exogenous} & OLS & 3 &
0.8969 & 0.8799 & 0.0170 \\
Free-curvature, Equation~\ref{eq-free-curvature-exogenous} & NLS & 3 &
0.8981 & 0.8817 & 0.0164 \\
\end{longtable}

Importantly, the evidence on acceleration does not depend on this
choice. Across \(k = 0, \ldots, 6\) the estimate of \(\beta\) moves by
at most \(0.022\) in the linearithmic specification and by at most
\(3.2 \times 10^{-5}\) in the quadratic specification, in both cases
about a fifth of its standard error, and remains positive at every lag
in every dataset, with HAC \(p\)-values no larger than \(0.013\).

\subsubsection*{Polynomial and free-curvature
specifications}\label{polynomial-and-free-curvature-specifications}
\addcontentsline{toc}{subsubsection}{Polynomial and free-curvature
specifications}

The free-curvature specification
Equation~\ref{eq-free-curvature-exogenous} is a reparametrization of the
general polynomial regression model,
\begin{equation}\protect\phantomsection\label{eq-nlls-regression}{
C_t = \alpha + \lambda t + \beta t^\delta + \gamma_1 N_{t-k} + \gamma_2 S_t + \varepsilon_t,
}\end{equation}

in which the degree of the trend is estimated jointly with its
coefficient by non-linear least squares \autocite{bates1988nonlinear},
with the HAC standard errors and the ENSO lag used throughout. Taken
directly to the expanding-window analysis,
Equation~\ref{eq-nlls-regression} is not identified at the two points
that matter most. At \(\delta = 1\) the regressor \(t^\delta\) is
exactly collinear with the linear trend, so \(\beta\) and \(\lambda\)
cannot be separated; at \(\delta = 0\) it is constant, so \(\beta\)
cannot be separated from the intercept. The first is precisely the
linearithmic case, the alternative of primary interest.

Parametrizing the curvature as
\(g(t,\delta) = (t^\delta - t)/(\delta - 1)\), as in
Equation~\ref{eq-free-curvature}, removes the collinearity. Dividing by
\(\delta - 1\) makes \(\delta = 1\) the linearithmic specification
exactly rather than a singularity to be excluded, while \(\delta = 2\)
gives \(t^2 - t\), the quadratic curvature over and above a linear
trend. Neither non-identified point falls within the estimation range,
and \(\delta \geq 1\) becomes a well-posed constraint with an
interpretable boundary: at least linearithmic. By condition 3 of
Theorem~\ref{thm-supralinear-trend} the sign of \(\beta\) determines
acceleration at every admissible \(\delta\), so the test requires no
separate inference on \(\hat\delta\), which is reported as a shape
estimate.

For the free-curvature specification, we impose \(\delta \in [1, 10]\)
on every fit\footnote{The upper bound is not what places Berkeley's
  exponent where it is. Re-estimating the full sample with the bound
  raised to \(20\), and again to \(50\), leaves the four interior
  estimates unchanged at \(3.39\), \(4.00\), \(4.08\) and \(4.69\), and
  moves Berkeley's to \(10.47\) under either ceiling. The constraint
  binds for one dataset, and only barely.}, including each bootstrap
replication, so that the simulated null distribution is built only from
fits that could have produced the observed statistic.

Figure~\ref{fig-nls-exponent-path} presents \(\hat{\delta}\) across
expanding windows with the start point fixed at 1970, and
Table~\ref{tbl-nls-final-window} reports the final window alongside the
linearithmic specification on the same sample.

\begin{longtable}[]{@{}
  >{\raggedright\arraybackslash}p{(\linewidth - 6\tabcolsep) * \real{0.4375}}
  >{\raggedleft\arraybackslash}p{(\linewidth - 6\tabcolsep) * \real{0.1875}}
  >{\raggedleft\arraybackslash}p{(\linewidth - 6\tabcolsep) * \real{0.1875}}
  >{\raggedleft\arraybackslash}p{(\linewidth - 6\tabcolsep) * \real{0.1875}}@{}}
\caption{Free-curvature specification
Equation~\ref{eq-free-curvature-exogenous} at the final expanding
window, with controls, compared with the linearithmic specification.
Bootstrap-adjusted \(p\)-values for \(H_0: \beta = 0\). Berkeley's
exponent is pinned at the upper
constraint.}\label{tbl-nls-final-window}\tabularnewline
\toprule\noalign{}
\begin{minipage}[b]{\linewidth}\raggedright
Dataset
\end{minipage} & \begin{minipage}[b]{\linewidth}\raggedleft
\(\hat\delta\)
\end{minipage} & \begin{minipage}[b]{\linewidth}\raggedleft
Adjusted \(p\), free-curvature
\end{minipage} & \begin{minipage}[b]{\linewidth}\raggedleft
Adjusted \(p\), linearithmic
\end{minipage} \\
\midrule\noalign{}
\endfirsthead
\toprule\noalign{}
\begin{minipage}[b]{\linewidth}\raggedright
Dataset
\end{minipage} & \begin{minipage}[b]{\linewidth}\raggedleft
\(\hat\delta\)
\end{minipage} & \begin{minipage}[b]{\linewidth}\raggedleft
Adjusted \(p\), free-curvature
\end{minipage} & \begin{minipage}[b]{\linewidth}\raggedleft
Adjusted \(p\), linearithmic
\end{minipage} \\
\midrule\noalign{}
\endhead
\bottomrule\noalign{}
\endlastfoot
GISTEMP & 4.00 & 0.050 & 0.010 \\
NOAA & 4.08 & 0.084 & 0.017 \\
ERA5 & 4.69 & 0.262 & 0.017 \\
HadCRUT & 3.39 & 0.213 & 0.119 \\
Berkeley & 10.00 & 0.751 & 0.246 \\
\end{longtable}

The two estimators agree on the timing of what they detect, but not on
how much is detectable. In the full sample the free-curvature
specification rejects only for GISTEMP, at exactly \(0.050\); NOAA
reaches \(0.084\), and ERA5, which the linearithmic specification
rejects at \(0.017\), reaches only \(0.262\). Across windows GISTEMP
rejects in \(22\) and NOAA in \(20\), first in the windows ending
2022-04 and 2022-01, while ERA5, HadCRUT and Berkeley never reject.
Within those windows the pattern runs opposite to the fixed-shape
specifications: NOAA's adjusted \(p\)-value falls to \(0.011\) at the
window ending 2023-06, stays below \(0.05\) through 2023-08 and exceeds
it in every later window, and GISTEMP's oscillates about the threshold
before returning to \(0.050\) at the final window.

The parameter estimates suggest one explanation. Between the windows
ending 2023-06 and 2025-12 both curvature parameters increase (NOAA's
\(\hat\delta\) from \(3.18\) to \(4.08\) and \(\hat\beta\) from \(0.62\)
to \(1.06\); GISTEMP's from \(2.06\) to \(4.00\) and from \(0.40\) to
\(1.00\)), while the standard error on \(\hat\beta\) rises faster, from
\(0.22\) to \(0.48\) and from \(0.16\) to \(0.42\). Since \(\beta\) and
\(\delta\) trade off against each other, the rise in \(\hat\delta\)
comes with a loss of precision on \(\hat\beta\). The opposite reading,
that the evidence weakens once the functional form is relaxed, is also
consistent with these numbers, and the estimates are too imprecise to
settle the question: over \(\delta \in [2,8]\) the residual sum of
squares of the final-window fit varies by at most \(2.2\%\) in every
dataset. What the paper claims does not turn on that question, since the
case for the linearithmic form rests on
Lemma~\ref{lem-linearithmic-speed} rather than on fit.

One caveat should be read alongside the figure: a share of windows sit
on a bound, where \(\hat\delta\) has no valid standard error, and the
two bounds mean opposite things. A window pinned at \(\delta = 1\) would
place the exponent below the linearithmic value if unconstrained, and
carries no curvature signal the constraint is hiding; these are confined
to the earlier record, \(199\) windows across the five datasets, none
ending after 2016-12. A window pinned at \(\delta = 10\) calls for more
curvature than the constraint allows; this is most common in Berkeley,
in \(123\) windows, and Berkeley is the only dataset whose final window
is pinned. Neither bound manufactures evidence: of the \(632\) window
fits resting on a constraint, not one rejects \(H_0: \beta = 0\).

\begin{figure}

\centering{

\includegraphics[width=0.85\linewidth,height=\textheight,keepaspectratio]{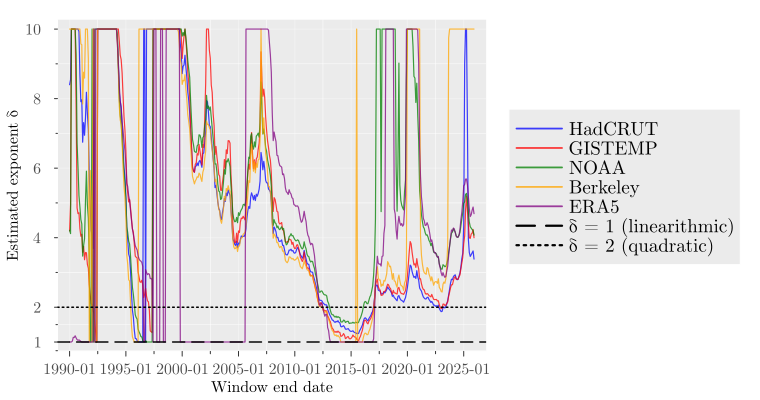}

}

\caption{\label{fig-nls-exponent-path}Estimated curvature exponent
\(\hat{\delta}\) from the free-curvature specification
Equation~\ref{eq-free-curvature-exogenous}, with El Niño and solar
controls, across expanding windows starting in 1970 and ending between
1990 and 2025. The dashed line marks \(\delta = 1\), the linearithmic
specification, and the dotted line \(\delta = 2\), quadratic curvature
over and above a linear trend. Estimates resting on the constraint
boundaries at \(\delta = 1\) and \(\delta = 10\) carry no standard error
and mean opposite things; they are discussed in the text.}

\end{figure}%

\subsubsection*{A grid of fixed curvature
exponents}\label{sec-delta-grid}
\addcontentsline{toc}{subsubsection}{A grid of fixed curvature
exponents}

The linearithmic and quadratic specifications fix the curvature exponent
at \(\delta = 1\) and \(\delta = 2\). This section repeats the
expanding-window procedure with \(\delta\) fixed at intermediate and
larger values, estimating

\begin{equation}\protect\phantomsection\label{eq-fixed-delta}{
x_t = \alpha + \lambda t + \beta\, t^{\delta} + \gamma_1 N_{t-k} + \gamma_2 S_t + \varepsilon_t
}\end{equation}

by OLS at each \(\delta\) on the grid. Every other setting matches the
main text: the El Niño and solar controls at the selected lag, the block
wild bootstrap with \(L = 10\) and \(1000\) replications, and the same
\(432\) nested windows, so the dates below are directly comparable with
those in Section~\ref{sec-onset}. The \(\delta = 2\) row reproduces the
stored quadratic results exactly, which puts the two procedures on the
same footing.

\begin{longtable}[]{@{}rlllll@{}}
\caption{First window in which \(H_0: \beta = 0\) is rejected
family-wise, by fixed curvature exponent, with the number of rejecting
windows out of \(432\) in parentheses. \(\delta = 1\) and \(\delta = 2\)
are the linearithmic and quadratic specifications reported in the main
text.}\label{tbl-delta-grid-onset}\tabularnewline
\toprule\noalign{}
\(\delta\) & GISTEMP & NOAA & ERA5 & HadCRUT & Berkeley \\
\midrule\noalign{}
\endfirsthead
\toprule\noalign{}
\(\delta\) & GISTEMP & NOAA & ERA5 & HadCRUT & Berkeley \\
\midrule\noalign{}
\endhead
\bottomrule\noalign{}
\endlastfoot
1.00 & 2023-09 (28) & 2023-07 (30) & 2024-07 (18) & never & never \\
1.25 & 2023-09 (28) & 2023-04 (33) & 2024-05 (20) & never & never \\
1.50 & 2023-09 (28) & 2023-03 (34) & 2024-05 (20) & never & never \\
1.75 & 2023-09 (28) & 2022-10 (35) & 2024-04 (21) & never & never \\
2.00 & 2023-10 (27) & 2022-06 (43) & 2024-04 (21) & never & never \\
2.50 & 2024-06 (19) & 2022-08 (37) & 2024-06 (19) & never & never \\
3.00 & 2024-10 (15) & 2023-03 (34) & 2024-09 (16) & never & never \\
\end{longtable}

\begin{longtable}[]{@{}rrrrrr@{}}
\caption{Bootstrap-adjusted \(p\)-values at the final window, 1970 to
2025, by fixed curvature
exponent.}\label{tbl-delta-grid-pvalues}\tabularnewline
\toprule\noalign{}
\(\delta\) & GISTEMP & NOAA & ERA5 & HadCRUT & Berkeley \\
\midrule\noalign{}
\endfirsthead
\toprule\noalign{}
\(\delta\) & GISTEMP & NOAA & ERA5 & HadCRUT & Berkeley \\
\midrule\noalign{}
\endhead
\bottomrule\noalign{}
\endlastfoot
1.00 & 0.010 & 0.017 & 0.017 & 0.119 & 0.246 \\
1.25 & 0.008 & 0.009 & 0.018 & 0.103 & 0.211 \\
1.50 & 0.007 & 0.007 & 0.017 & 0.095 & 0.184 \\
1.75 & 0.006 & 0.004 & 0.015 & 0.094 & 0.176 \\
2.00 & 0.004 & 0.003 & 0.013 & 0.085 & 0.155 \\
2.50 & 0.010 & 0.003 & 0.017 & 0.098 & 0.147 \\
3.00 & 0.010 & 0.002 & 0.018 & 0.115 & 0.150 \\
\end{longtable}

Three features are worth recording. The set of rejecting datasets does
not change anywhere on the grid: GISTEMP, NOAA and ERA5 reject at every
exponent, HadCRUT and Berkeley at none, so the three-of-five conclusion
is not a consequence of the shape assumed. The first rejection moves by
at most about a year within each dataset, and it moves
non-monotonically: the smallest adjusted \(p\)-values fall at
intermediate exponents in every dataset except NOAA, whose smallest is
at \(\delta = 3\). Relaxing the exponent below the quadratic brings
detection forward by one month for GISTEMP and pushes it back for NOAA
and ERA5.

\subsubsection*{Results without the El Niño and solar
controls}\label{sec-no-controls}
\addcontentsline{toc}{subsubsection}{Results without the El Niño and
solar controls}

The main text reports every specification with the El Niño and solar
controls included, following \textcite{Forster2026}. This section
reports the same specifications estimated without them.

Table~\ref{tbl-results-no-controls} gives the full-sample estimates. The
curvature coefficient \(\hat\beta\) remains positive in all five
datasets in all three specifications, so the direction of the evidence
is unchanged. Its significance is not. Under the linearithmic
specification the HAC \(p\)-value falls below \(5\%\) in three of the
five datasets (NOAA \(0.016\), GISTEMP \(0.024\), ERA5 \(0.045\)) and
not in HadCRUT (\(0.113\)) or Berkeley (\(0.166\)); under the quadratic
specification the same three reject (NOAA \(0.004\), GISTEMP \(0.010\),
ERA5 \(0.022\)) with HadCRUT (\(0.061\)) and Berkeley (\(0.086\))
marginal; under the free-curvature specification none rejects, the
\(p\)-values ranging from \(0.198\) to \(0.447\). The free-curvature
fits also carry two symptoms of the controls being absent rather than of
curvature being present: the exponent rests on the upper constraint in
HadCRUT and Berkeley and lies between \(6.9\) and \(9.3\) in the other
three, against interior values of \(3.4\) to \(4.7\) once the controls
are included, and the standard errors on \(\hat\beta\) are \(13\) to
\(16\) times the linearithmic standard errors in the three datasets
whose exponent becomes interior (GISTEMP, NOAA and ERA5). With the ENSO and solar variation left in the residual the
objective is flatter in \(\delta\), so the optimiser drifts toward the
boundary and the ridge along which \(\beta\) and \(\delta\) trade off
widens.

\begin{longtable}[]{@{}lrrrr@{}}
\caption{Estimation results for the three specifications estimated on
1970--2025 without the El Niño and solar controls. Standard errors,
calculated using a HAC estimator, are shown in parentheses. † marks an
exponent resting on the upper constraint \(\delta = 10\), for which no
standard error is
defined.}\label{tbl-results-no-controls}\tabularnewline
\toprule\noalign{}
\textbf{Dataset / Model} & \(\alpha\) & \(\beta\) & \(\lambda\) &
\(\delta\) \\
\midrule\noalign{}
\endfirsthead
\toprule\noalign{}
\textbf{Dataset / Model} & \(\alpha\) & \(\beta\) & \(\lambda\) &
\(\delta\) \\
\midrule\noalign{}
\endhead
\bottomrule\noalign{}
\endlastfoot
\textbf{Berkeley} & & & & \\
Linearithmic & 0.2312 & 0.2138 & -0.2756 & \\
& (0.0536) & (0.1543) & (1.0119) & \\
Quadratic & 0.2202 & 0.0005 & 0.7899 & \\
& (0.0425) & (0.0003) & (0.2026) & \\
Free-curvature & 0.1973 & 2.5014 & 1.3022 & 10.00\textsuperscript{†} \\
& (0.0305) & (2.6088) & (0.1342) & --- \\
\textbf{ERA5} & & & & \\
Linearithmic & 0.2635 & 0.3417 & -1.0906 & \\
& (0.0640) & (0.1707) & (1.1250) & \\
Quadratic & 0.2383 & 0.0008 & 0.6587 & \\
& (0.0542) & (0.0003) & (0.2325) & \\
Free-curvature & 0.1994 & 2.8309 & 1.3686 & 9.1895 \\
& (0.0403) & (2.6987) & (0.1547) & (5.8644) \\
\textbf{GISTEMP} & & & & \\
Linearithmic & 0.2963 & 0.3412 & -1.1499 & \\
& (0.0505) & (0.1515) & (0.9915) & \\
Quadratic & 0.2692 & 0.0008 & 0.6085 & \\
& (0.0389) & (0.0003) & (0.1914) & \\
Free-curvature & 0.2278 & 2.6269 & 1.2927 & 9.2890 \\
& (0.0283) & (2.0387) & (0.1012) & (5.3137) \\
\textbf{HadCRUT} & & & & \\
Linearithmic & 0.2651 & 0.2329 & -0.4137 & \\
& (0.0506) & (0.1470) & (0.9651) & \\
Quadratic & 0.2494 & 0.0006 & 0.7692 & \\
& (0.0397) & (0.0003) & (0.1907) & \\
Free-curvature & 0.2207 & 2.1941 & 1.2739 & 10.00\textsuperscript{†} \\
& (0.0299) & (2.8858) & (0.1437) & --- \\
\textbf{NOAA} & & & & \\
Linearithmic & 0.2332 & 0.3491 & -1.2422 & \\
& (0.0529) & (0.1453) & (0.9587) & \\
Quadratic & 0.2072 & 0.0008 & 0.5459 & \\
& (0.0403) & (0.0003) & (0.1853) & \\
Free-curvature & 0.1707 & 1.8884 & 1.2198 & 6.8869 \\
& (0.0313) & (1.9191) & (0.1427) & (4.7308) \\
\end{longtable}

The expanding-window evidence is unambiguous.
Figure~\ref{fig-bootstrap-pvals-no-controls} gives the
bootstrap-adjusted \(p\)-values across the \(432\) windows. Not one of
the \(2{,}160\) window-level tests rejects \(H_0: \beta = 0\) at the
\(5\%\) level in either OLS specification: the smallest adjusted
\(p\)-value is \(0.377\) for the linearithmic specification and
\(0.281\) for the quadratic. The nominal full-sample significance
reported above therefore does not survive the correction for multiple
testing, which is the same conclusion the main text reaches for HadCRUT
and Berkeley with the controls included, reached here for every dataset.

\begin{figure}

\begin{minipage}[t]{0.50\linewidth}

\raisebox{-\height}{

\includegraphics[width=1\linewidth,height=\textheight,keepaspectratio,alt={Linearithmic specification without additional covariates.}]{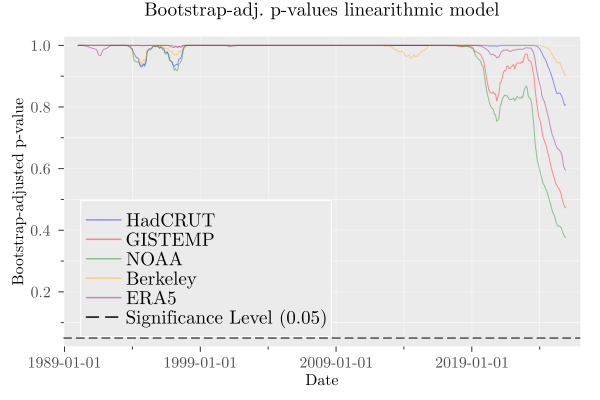}

}

\subcaption{\label{}Linearithmic specification without additional
covariates.}
\end{minipage}%
\begin{minipage}[t]{0.50\linewidth}

\raisebox{-\height}{

\includegraphics[width=1\linewidth,height=\textheight,keepaspectratio,alt={Quadratic specification without additional covariates.}]{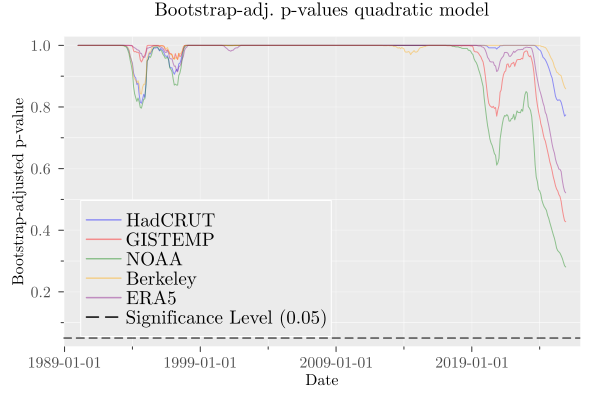}

}

\subcaption{\label{}Quadratic specification without additional
covariates.}
\end{minipage}%

\caption{\label{fig-bootstrap-pvals-no-controls}Bootstrap-adjusted
\(p\)-values for \(H_0: \beta = 0\) across expanding windows ending
between 1990 and 2025, for the linearithmic and quadratic specifications
estimated without the El Niño and solar controls. The dashed line marks
the 5\% level; no series crosses it. The free-curvature specification is
not bootstrapped without controls.}

\end{figure}%

The free-curvature specification is not carried through the
expanding-window bootstrap without controls. Estimated without them the
exponent rests on or near the upper constraint in every dataset, as
Table~\ref{tbl-results-no-controls} shows, so \(\hat\delta\) would
record which boundary the optimiser reached rather than the shape of the
trend, and the full-sample test does not reject in any dataset to begin
with. The two OLS specifications, which do reject nominally in three
datasets before the correction, are the informative case here, and
neither survives it.

Figure~\ref{fig-fitted-values-no-controls} gives the corresponding
fitted values.

\begin{figure}

\begin{minipage}[t]{0.33\linewidth}

\raisebox{-\height}{

\includegraphics[width=1\linewidth,height=\textheight,keepaspectratio,alt={Linearithmic fit.}]{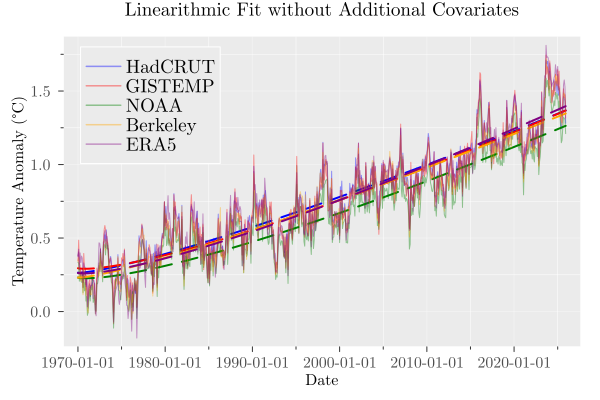}

}

\subcaption{\label{}Linearithmic fit.}
\end{minipage}%
\begin{minipage}[t]{0.33\linewidth}

\raisebox{-\height}{

\includegraphics[width=1\linewidth,height=\textheight,keepaspectratio,alt={Quadratic fit.}]{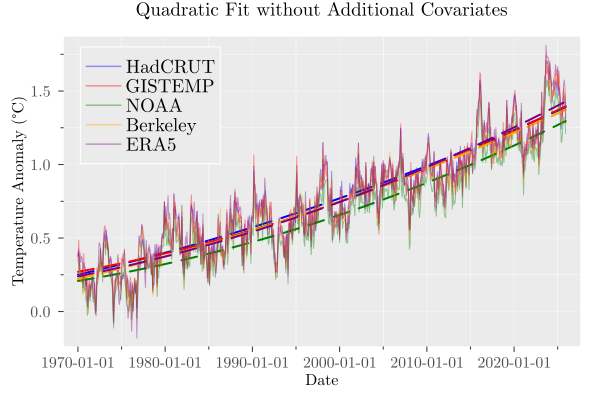}

}

\subcaption{\label{}Quadratic fit.}
\end{minipage}%
\begin{minipage}[t]{0.33\linewidth}

\raisebox{-\height}{

\includegraphics[width=1\linewidth,height=\textheight,keepaspectratio,alt={Free-curvature fit.}]{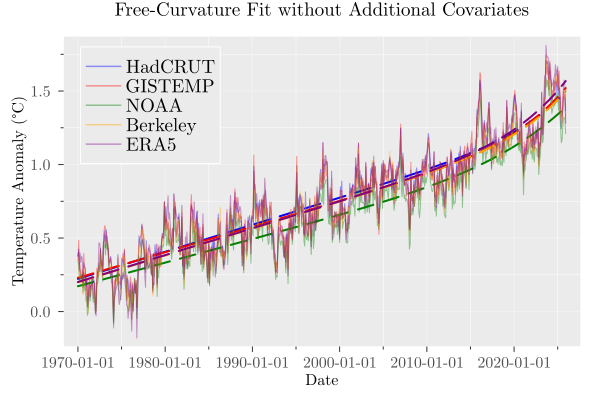}

}

\subcaption{\label{}Free-curvature fit.}
\end{minipage}%

\caption{\label{fig-fitted-values-no-controls}Actual and fitted values
for the three specifications from 1970 to 2025, estimated without the El
Niño and solar controls.}

\end{figure}%

\subsubsection*{Projections under alternative trend
specifications}\label{sec-projections}
\addcontentsline{toc}{subsubsection}{Projections under alternative trend
specifications}

This section gives the construction behind Figure~\ref{fig-projections}
and the results in full.

Each fitted trend is extended from the end of the sample to 2050-12, and
two quantities are evaluated: the implied warming rate, \(d h_t/dt\), in
°C decade\(^{-1}\), and the projected level of the anomaly. Two features
of the construction matter. First, the trend regressors are scaled by
the estimation sample length \(T\), so the projections hold \(T\) fixed
at its in-sample value and let \(t\) run beyond it; rescaling by the
extended length would alter the model rather than extend it. Second, the
El Niño and solar controls are not forecast but held at their in-sample
means, so they contribute a constant level shift and nothing at all to
the projected warming rate. The exercise is therefore a comparison
between functional forms, not a projection in the sense of the scenario
literature.

\begin{longtable}[]{@{}lrrr@{}}
\caption{Implied warming rate (°C decade\(^{-1}\)) at the end of the
sample and in 2050, and projected anomaly level (°C) in 2050, under each
specification with controls. Each cell gives the range across the five
temperature datasets; the free-curvature row covers the four datasets
whose exponent is interior, Berkeley being pinned at the upper
constraint.}\label{tbl-projections}\tabularnewline
\toprule\noalign{}
Specification & Rate, 2025 & Rate, 2050 & Level, 2050 \\
\midrule\noalign{}
\endfirsthead
\toprule\noalign{}
Specification & Rate, 2025 & Rate, 2050 & Level, 2050 \\
\midrule\noalign{}
\endhead
\bottomrule\noalign{}
\endlastfoot
Linear & 0.196--0.216 & 0.196--0.216 & 1.70--1.89 \\
Linearithmic & 0.251--0.281 & 0.271--0.310 & 1.96--2.16 \\
Quadratic & 0.283--0.324 & 0.351--0.420 & 2.18--2.38 \\
Free-curvature & 0.326--0.451 & 0.555--1.262 & 2.50--3.48 \\
\end{longtable}

Table~\ref{tbl-projections} reports the ranges across datasets. The
linearithmic rate at the end of the sample, \(0.251\) to \(0.281\) °C
decade\(^{-1}\), reproduces the figure quoted in
Section~\ref{sec-results}, and the linear specification holds its rate
fixed by construction, so the difference between the two columns for any
other specification is the contribution of the estimated curvature
alone.

The free-curvature specification is treated differently, in two
respects. It is non-linear in \(\delta\) and no inference is conducted
on \(\hat\delta\), so no band is computed for it; and Berkeley, whose
exponent is pinned at the upper constraint of \(10\), is excluded from
the plotted line. That exclusion is worth being explicit about, because
the excluded case is informative. Applying the same calculation to the
pinned Berkeley fit returns a warming rate of \(12.5\) °C
decade\(^{-1}\) and an anomaly of \(11.6\) °C by 2050, which is ten to
twenty times the rate and three to five times the anomaly of the other
four datasets. An exponent resting on a bound carries no standard error
and should not be extrapolated; this is the same reason the inference in
this paper rests on \(\hat\beta\) throughout and \(\hat\delta\) is
reported only as a shape estimate.

\subsubsection*{Size of the test under a linear null}\label{sec-size}
\addcontentsline{toc}{subsubsection}{Size of the test under a linear
null}

HAC inference is known to over-reject in finite samples when the errors
are strongly persistent, and monthly GMST residuals are persistent.
Because the whole paper rests on a \(t\)-test of \(H_0: \beta = 0\), the
size of that test under the null it actually tests is worth establishing
directly rather than assuming.

The simulation uses the paper's own design matrix, not a stylised
analogue. The regressors are those of
Equation~\ref{eq-linearithmic-trend-exogenous}, including the lagged
ENSO and solar controls, over the estimation sample of \(T = 672\)
months. Data are generated under a \emph{pure linear} trend, so the null
is true by construction, with AR(1) errors of persistence \(\phi\), and
\(H_0: \beta = 0\) is tested with the same HAC \(t\)-statistic used
throughout the paper. Table~\ref{tbl-size} reports rejection frequencies
over \(5000\) replications.

\begin{longtable}[]{@{}lrr@{}}
\caption{Rejection frequency of the HAC \(t\)-test of \(H_0: \beta = 0\)
when the true trend is linear, so that the null holds by construction.
Design matrix as estimated in the paper, \(T = 672\), AR(1) errors,
\(5000\) replications. Entries are actual sizes; correct sizes would be
\(0.05\) and \(0.01\).}\label{tbl-size}\tabularnewline
\toprule\noalign{}
\textbf{\(\phi\)} & Nominal 5\% & Nominal 1\% \\
\midrule\noalign{}
\endfirsthead
\toprule\noalign{}
\textbf{\(\phi\)} & Nominal 5\% & Nominal 1\% \\
\midrule\noalign{}
\endhead
\bottomrule\noalign{}
\endlastfoot
0.00 & 0.075 & 0.023 \\
0.50 & 0.101 & 0.034 \\
0.90 & 0.208 & 0.108 \\
0.95 & 0.310 & 0.193 \\
\end{longtable}

The test is over-sized, and increasingly so as persistence rises: at
\(\phi = 0\) the nominal \(5\%\) test rejects \(7.5\%\) of the time, and
by \(\phi = 0.95\) it rejects \(31\%\) of the time. This is a property
of HAC inference in samples of this length rather than of the
linearithmic specification, but it means that the raw HAC \(p\)-values
reported in Table~\ref{tbl-results-full-sample} overstate the evidence
and should not be read at face value.

We therefore size-correct, modelling the errors of each dataset as
ARMA(1,1). The residual autocorrelations decay slowly; for HadCRUT they
are \(0.55\), \(0.43\) and \(0.11\) at lags 1, 2 and 12. An ARMA(1,1)
model captures this decay. Fitted by maximum likelihood to each
dataset's own residuals, it gives \(\hat\phi\) between \(0.72\) and
\(0.79\) and \(\hat\theta\) between \(-0.35\) and \(-0.30\), and it is
preferred to AR(1) by \(14\) to \(21\) BIC points in every dataset. The
null distribution of \(|t|\) is then simulated under the fitted errors
over \(20000\) replications, and the observed statistic is referred to
that distribution rather than to the normal.
Table~\ref{tbl-size-corrected} reports the result.

\begin{longtable}[]{@{}
  >{\raggedright\arraybackslash}p{(\linewidth - 12\tabcolsep) * \real{0.2800}}
  >{\raggedleft\arraybackslash}p{(\linewidth - 12\tabcolsep) * \real{0.1200}}
  >{\raggedleft\arraybackslash}p{(\linewidth - 12\tabcolsep) * \real{0.1200}}
  >{\raggedleft\arraybackslash}p{(\linewidth - 12\tabcolsep) * \real{0.1200}}
  >{\raggedleft\arraybackslash}p{(\linewidth - 12\tabcolsep) * \real{0.1200}}
  >{\raggedleft\arraybackslash}p{(\linewidth - 12\tabcolsep) * \real{0.1200}}
  >{\raggedleft\arraybackslash}p{(\linewidth - 12\tabcolsep) * \real{0.1200}}@{}}
\caption{Size-corrected inference for the linearithmic and quadratic
coefficients on the full sample with controls. \(\hat\phi\) and
\(\hat\theta\) are the autoregressive and moving-average parameters of
an ARMA(1,1) model fitted to each dataset's residuals from the
linearithmic specification, and the same errors are used for both
specifications; critical values and \(p\)-values are simulated under a
linear null with those errors, \(20000\) replications. The uncorrected
critical values would be \(1.96\) and
\(2.58\).}\label{tbl-size-corrected}\tabularnewline
\toprule\noalign{}
\begin{minipage}[b]{\linewidth}\raggedright
\textbf{Specification / Dataset}
\end{minipage} & \begin{minipage}[b]{\linewidth}\raggedleft
\(\hat\phi\)
\end{minipage} & \begin{minipage}[b]{\linewidth}\raggedleft
\(\hat\theta\)
\end{minipage} & \begin{minipage}[b]{\linewidth}\raggedleft
Critical value, 5\%
\end{minipage} & \begin{minipage}[b]{\linewidth}\raggedleft
Critical value, 1\%
\end{minipage} & \begin{minipage}[b]{\linewidth}\raggedleft
Observed \(|t|\)
\end{minipage} & \begin{minipage}[b]{\linewidth}\raggedleft
Size-corrected \(p\)
\end{minipage} \\
\midrule\noalign{}
\endfirsthead
\toprule\noalign{}
\begin{minipage}[b]{\linewidth}\raggedright
\textbf{Specification / Dataset}
\end{minipage} & \begin{minipage}[b]{\linewidth}\raggedleft
\(\hat\phi\)
\end{minipage} & \begin{minipage}[b]{\linewidth}\raggedleft
\(\hat\theta\)
\end{minipage} & \begin{minipage}[b]{\linewidth}\raggedleft
Critical value, 5\%
\end{minipage} & \begin{minipage}[b]{\linewidth}\raggedleft
Critical value, 1\%
\end{minipage} & \begin{minipage}[b]{\linewidth}\raggedleft
Observed \(|t|\)
\end{minipage} & \begin{minipage}[b]{\linewidth}\raggedleft
Size-corrected \(p\)
\end{minipage} \\
\midrule\noalign{}
\endhead
\bottomrule\noalign{}
\endlastfoot
\textbf{Linearithmic} & & & & & & \\
HadCRUT & 0.755 & \(-0.304\) & 2.55 & 3.42 & 3.17 & 0.0163 \\
GISTEMP & 0.723 & \(-0.300\) & 2.50 & 3.37 & 4.22 & 0.0018 \\
NOAA & 0.774 & \(-0.326\) & 2.59 & 3.47 & 4.24 & 0.0023 \\
Berkeley & 0.779 & \(-0.352\) & 2.59 & 3.48 & 2.85 & 0.0314 \\
ERA5 & 0.787 & \(-0.302\) & 2.62 & 3.51 & 4.34 & 0.0022 \\
\textbf{Quadratic} & & & & & & \\
HadCRUT & 0.755 & \(-0.304\) & 2.48 & 3.35 & 3.66 & 0.0055 \\
GISTEMP & 0.723 & \(-0.300\) & 2.43 & 3.27 & 4.94 & 0.0003 \\
NOAA & 0.774 & \(-0.326\) & 2.52 & 3.40 & 5.50 & \(<0.0001\) \\
Berkeley & 0.779 & \(-0.352\) & 2.52 & 3.41 & 3.39 & 0.0103 \\
ERA5 & 0.787 & \(-0.302\) & 2.55 & 3.43 & 4.79 & 0.0006 \\
\end{longtable}

Two points follow. First, the correction is substantial: under the
fitted ARMA(1,1) errors the nominal \(5\%\) test rejects a true null
\(12\) to \(14\%\) of the time, and the size-corrected \(5\%\) critical
value is \(2.5\) to \(2.6\) rather than \(1.96\), and the \(1\%\) value
\(3.4\) to \(3.5\) rather than \(2.58\). The error model matters too.
Critical values calibrated under AR(1) errors would be \(2.38\) to
\(2.42\), and a test using them has an actual size of \(6.1\) to
\(6.7\%\) under the fitted ARMA(1,1) errors. Second, the finding
survives the correction. Every dataset rejects \(H_0: \beta = 0\) at the
size-corrected \(5\%\) level, with corrected \(p\)-values between
\(0.0018\) and \(0.0314\). Berkeley's margin is the narrowest, with
\(|t| = 2.85\) against a critical value of \(2.59\). At the
size-corrected \(1\%\) level three of the five reject, GISTEMP, NOAA and
ERA5, while HadCRUT and Berkeley do not: their statistics of \(3.17\)
and \(2.85\) fall below the corrected critical values of \(3.42\) and
\(3.48\). The quadratic coefficient, size-corrected in the same way,
gives the same answer at \(5\%\): all five datasets reject, with
corrected \(p\)-values ranging from below \(0.0001\) to \(0.0103\). Its
corrections are slightly smaller, with \(5\%\) critical values of
\(2.43\) to \(2.55\). At \(1\%\) four of the five reject, Berkeley
falling just short with \(|t| = 3.39\) against a critical value of
\(3.41\).

This size distortion is the reason the multiple-testing procedure of the
following section resamples in blocks. A wild bootstrap with independent
Rademacher multipliers preserves heteroskedasticity but not serial
correlation, so its pseudo-samples are less persistent than the data and
the null distribution it builds is too tight. Multiple-testing
correction does not repair this: the two are separate problems, and an
i.i.d. wild bootstrap would leave the expanding-window \(p\)-values
carrying the same distortion as the raw ones. The next section reports
the block-length calibration and quantifies the distortion that remains.

\subsubsection*{Bootstrap adjustment for multiple
testing}\label{bootstrap-adjustment-for-multiple-testing}
\addcontentsline{toc}{subsubsection}{Bootstrap adjustment for multiple
testing}

To control for multiple testing across expanding windows, we use a
max-\(|t|\) wild-bootstrap algorithm, following the resampling-based
family-wise error control logic of \textcite{westfall1993resampling}.
Let the starting date be fixed at \(\tau_0\), and let the set of window
end points be \(\mathcal{T}=\{\tau_1,\ldots,\tau_m\}\), with window
\(j\) given by \([\tau_0,\tau_j]\). For each window, we test the null of
no acceleration, \(H_{0,j}:\beta_j=0\), where \(\beta_j\) is the
acceleration coefficient of interest (linearithmic or polynomial
coefficient, depending on the specification).

For each window \(j\), we estimate: (i) the unrestricted model,
obtaining the HAC \(t\)-statistic \(t_j\), and (ii) the restricted model
under \(H_{0,j}\), obtaining fitted values \(\widehat{C}^{(0)}_{t,j}\)
and residuals \(\widehat{u}^{(0)}_{t,j}\). The restricted residuals are
centred before resampling. We then generate bootstrap pseudo-series with
Rademacher multipliers:

\[
C^*_{t,j}=\widehat{C}^{(0)}_{t,j}+\widehat{u}^{(0)}_{t,j}\,\xi_t,
\]

where the multipliers \(\xi_t\) take values in \(\{-1,+1\}\) with equal
probability,

\[
\xi_t\in\{-1,+1\},\;\Pr(\xi_t=1)=\Pr(\xi_t=-1)=\tfrac12,
\]

and are held constant within blocks of \(L\) consecutive months rather
than drawn independently at each date. Writing
\(\xi_t = \zeta_{\lceil t/L \rceil}\) with \(\zeta_1,\zeta_2,\ldots\)
independent Rademacher draws, the sign applied to the residuals changes
only at block boundaries, so the serial correlation of
\(\widehat{u}^{(0)}_{t,j}\) within a block is carried into the
pseudo-series. Setting \(L=1\) recovers the i.i.d. wild bootstrap
exactly, including the random number stream, so that case serves as a
regression test rather than an approximation. The same multiplier draw
is used across windows within each replication and truncated to each
window length.

At bootstrap replication \(b=1,\ldots,B\), we re-estimate the
unrestricted model on \(C^*_{t,j}\) for all \(j\in\{1,\ldots,m\}\) and
compute HAC \(t\)-statistics \(t^*_{j,b}\). The max-\(|t|\) statistic
is:

\[
M^*_b = \max_{1\leq j\leq m} \left| t^*_{j,b} \right|.
\]

The family-wise adjusted \(p\)-value for window \(j\) is computed as:

\[
\widehat{p}^{\,adj}_j = \frac{1+\sum_{b=1}^{B}\mathbf{1}\{M^*_b\geq |t_j|\}}{B+1},
\]

where \(\mathbf{1}\{\cdot\}\) is the indicator function. Hence, a
rejection at level \(\alpha\) occurs when
\(\widehat{p}^{\,adj}_j\leq \alpha\) or, equivalently, when
\(|t_j|>c_{1-\alpha}\), where \(c_{1-\alpha}\) is the empirical
\((1-\alpha)\) quantile of \(\{M^*_b\}_{b=1}^B\). In the empirical
application, we set \(B=1000\) and \(\alpha=0.05\).

\paragraph*{Calibration of the block
length}\label{calibration-of-the-block-length}
\addcontentsline{toc}{paragraph}{Calibration of the block length}

The block length is a tuning parameter and was selected against the
null, before examining its effect on any reported result. We simulate
under a pure linear trend on the paper's design matrix, so that
\(H_0:\beta=0\) holds by construction, draw ARMA(1,1) errors with the
parameters fitted to each dataset's residuals
(Table~\ref{tbl-size-corrected}), and record the proportion of
replications in which the max-\(|t|\) procedure rejects at a nominal
\(5\%\). A correctly sized procedure returns \(0.05\). The standard
error at that value is \(0.0069\) with \(1000\) Monte Carlo replications
and \(0.0034\) with \(4000\).

\begin{longtable}[]{@{}llllll@{}}
\caption{Actual size of the max-\(|t|\) block wild bootstrap at a
nominal \(5\%\) level, under a linear null with ARMA(1,1) errors fitted
to each dataset's residuals, by block length \(L\). \(999\) bootstrap
replications each; \(1000\) Monte Carlo replications for \(L = 1\) and
\(5\), and \(4000\) for \(L \geq 10\). \(L=1\) is the i.i.d. wild
bootstrap. The selected value is in
bold.}\label{tbl-block-calibration}\tabularnewline
\toprule\noalign{}
\(L\) & HadCRUT & GISTEMP & NOAA & Berkeley & ERA5 \\
\midrule\noalign{}
\endfirsthead
\toprule\noalign{}
\(L\) & HadCRUT & GISTEMP & NOAA & Berkeley & ERA5 \\
\midrule\noalign{}
\endhead
\bottomrule\noalign{}
\endlastfoot
1 & 0.109 & 0.094 & 0.100 & 0.105 & 0.108 \\
5 & 0.076 & 0.073 & 0.069 & 0.068 & 0.081 \\
\textbf{10} & \textbf{0.064} & \textbf{0.062} & \textbf{0.059} &
\textbf{0.063} & \textbf{0.064} \\
15 & 0.057 & 0.059 & 0.057 & 0.060 & 0.058 \\
20 & 0.058 & 0.061 & 0.060 & 0.057 & 0.058 \\
25 & 0.056 & 0.062 & 0.062 & 0.058 & 0.059 \\
\end{longtable}

The i.i.d. bootstrap is clearly over-sized, rejecting a true null
\(9.4\) to \(10.9\%\) of the time against a nominal \(5\%\). Blocking
removes most of that. At the selected \(L = 10\) the actual size is
\(5.9\) to \(6.4\%\), so the procedure remains mildly over-sized under
these errors, by about one percentage point. Longer blocks do little to
close the gap: size is essentially flat from \(L = 15\) onwards, at
\(5.6\) to \(6.2\%\), so no block length restores exactly nominal size,
and the gain from lengthening the block beyond \(10\) months is about
half a percentage point. We therefore retain \(L = 10\). Adjusted
\(p\)-values close to \(0.05\) should be read with this residual
over-rejection in mind.

We also re-ran the entire expanding-window analysis with the shorter
blocks \(L = 1\) and \(L = 5\). The same three datasets, GISTEMP, NOAA
and ERA5, reject in the full sample under both the linearithmic and the
quadratic specification at each of \(L = 1\), \(5\) and \(10\), with
final-window adjusted \(p\)-values between \(0.010\) and \(0.022\) for
the linearithmic specification and between \(0.002\) and \(0.020\) for
the quadratic. The first family-wise rejection moves by less than a year
in every case. HadCRUT and Berkeley reject in no window at any of these
block lengths. Their adjusted \(p\)-values do rise with \(L\), from
\(0.085\) and \(0.137\) at \(L=1\) to \(0.119\) and \(0.246\) at
\(L=10\), which is the expected direction: longer blocks retain more of
the persistence and widen the simulated null. Across \(L = 1\), \(5\)
and \(10\), the dating and the set of datasets carrying the result are
the same.

\paragraph*{Sensitivity to the start of the sample}\label{sec-startdate}
\addcontentsline{toc}{paragraph}{Sensitivity to the start of the sample}

Because \(t\log(t)\) is not invariant to a shift of the time origin, the
choice of 1970 is a modelling assumption rather than a normalisation,
and it needs checking. We re-estimated the expanding-window analysis
from five alternative anchors, holding the first window end at 1990
throughout so that every anchor is tested on the same \(432\) end points
and the adjusted \(p\)-values remain comparable.

\begin{longtable}[]{@{}lllllll@{}}
\caption{Full-sample bootstrap-adjusted \(p\)-values for \(H_0:\beta=0\)
in the linearithmic specification with controls, by the anchor at which
the sample begins. The published anchor is in bold. Computed under the
block wild bootstrap with \(L = 10\), as in the main
results.}\label{tbl-startdate}\tabularnewline
\toprule\noalign{}
Anchor & \(T\) & GISTEMP & NOAA & ERA5 & HadCRUT & Berkeley \\
\midrule\noalign{}
\endfirsthead
\toprule\noalign{}
Anchor & \(T\) & GISTEMP & NOAA & ERA5 & HadCRUT & Berkeley \\
\midrule\noalign{}
\endhead
\bottomrule\noalign{}
\endlastfoot
1955 & 852 & 0.001 & 0.001 & 0.001 & 0.001 & 0.001 \\
1960 & 792 & 0.003 & 0.002 & 0.001 & 0.003 & 0.013 \\
1965 & 732 & 0.041 & 0.016 & 0.010 & 0.070 & 0.181 \\
\textbf{1970} & 672 & 0.010 & 0.017 & 0.017 & 0.119 & 0.246 \\
1975 & 612 & 0.123 & 0.036 & 0.393 & 0.349 & 0.367 \\
1980 & 552 & 0.026 & 0.020 & 0.018 & 0.385 & 0.424 \\
\end{longtable}

The set of rejecting datasets depends on the anchor. From 1955 all five
datasets reject at the \(1\%\) level, and from 1960 all five reject at
\(5\%\), four of them at \(1\%\). From 1965 the same three datasets
reject at \(5\%\) as at the published 1970 anchor, GISTEMP, NOAA and
ERA5, with GISTEMP only narrowly. From 1975 only NOAA rejects, and from
1980 GISTEMP, NOAA and ERA5 do. The evidence is therefore not monotone
in the anchor, and the results depend on its choice. We adopt 1970
because it aligns with the observed increase in the warming rate
documented in the literature
\autocite{gadea-rivas_trends_2024,estrada_2013}. The anchor also affects
the dating. From the 1955 and 1960 anchors the first rejections arrive
between 1990 and 2007, against 2023--2024 from 1970. Samples that begin
in the mid-century period of little warming contain the subsequent rise
in the warming rate, so their rejections need not reflect the recent
curvature that the 1970 analysis detects. The onset dates in
Section~\ref{sec-onset} should therefore be read as conditional on the
1970 anchor, as should any comparison of them with dates in the
literature obtained from different samples.

\end{document}